\documentclass[12pt]{article}

\usepackage[english]{babel}
\usepackage[utf8x]{inputenc}
\usepackage{amsmath}
\usepackage{amsthm}
\usepackage{amsfonts}
\usepackage{algorithm}
\usepackage{algorithmic}
\usepackage{graphicx}
\usepackage[colorinlistoftodos]{todonotes}
\usepackage{fullpage}
\usepackage{bbm}
\usepackage{wrapfig}
\usepackage{subcaption}
\usepackage[normalem]{ulem}

\newtheorem{theorem}{Theorem}

\newtheorem{definition}[theorem]{Definition}
\newtheorem{lemma}[theorem]{Lemma}
\newtheorem{corollary}[theorem]{Corollary}

\newcommand{\R}{\mathbb{R}}
\newcommand{\eps}{\varepsilon}

\DeclareMathOperator{\Gr}{Gr}
\DeclareMathOperator{\E}{\mathbb{E}}

\def \Reg  {{\sf Reg}}
\def \Vol  {{\sf Vol}}

\def \wid {{\sf width}}
\def \mid {{\sf mid}}
\def \bkt {{\sf bkt}}
\def \poly {{\sf poly}}
\def \area {{\sf Area}}
\def \perim {{\sf Perimeter}}
\def \Conv {{\sf Conv}}
\def \one {{\bf 1}}
\renewcommand{\dot}[2]{{\langle {#1}, {#2}\rangle}}
\newcommand{\abs}[1]{{\left\vert {#1} \right\vert}}
\newcommand{\norm}[1]{{\left\Vert {#1} \right\Vert}}

\title{Contextual Search via Intrinsic Volumes}
\author{Renato Paes Leme \\ Google Research \and Jon Schneider \\ Princeton
University}
\date{}
\begin{document}

\maketitle

\begin{abstract}
We study the problem of contextual search, a multidimensional generalization
  of binary search that captures many problems in contextual decision-making. In
  contextual search, a learner is trying to learn the value of a hidden vector
  $v \in [0,1]^d$. Every round the learner is provided an adversarially-chosen
  context $u_t \in \R^d$, submits a guess $p_t$ for the value of $\langle u_t,
  v\rangle$, learns whether $p_t < \langle u_t, v\rangle$, and incurs loss
  $\ell(\langle u_t, v\rangle, p_t)$ (for some loss function $\ell$). The learner's goal is to minimize their total loss
  over the course of $T$ rounds.

We present an algorithm for the contextual search problem for the symmetric
  loss function $\ell(\theta, p) = |\theta - p|$ that achieves $O_{d}(1)$ total
  loss. We present a new algorithm for the dynamic pricing problem (which can be
  realized as a special case of the contextual search problem) that achieves
  $O_{d}(\log \log T)$ total loss, improving on the previous best known upper
  bounds of $O_{d}(\log T)$ and matching the known lower bounds (up to a
  polynomial dependence on $d$). Both algorithms make significant use of ideas
  from the field of integral geometry, most notably the notion of intrinsic
  volumes of a convex set. To the best of our knowledge this is the first
  application of intrinsic volumes to algorithm design. 
\end{abstract}

\section{Introduction}

Consider the classical problem of binary search, where the goal is to find a
hidden real number $x \in [0,1]$, and where feedback is limited to guessing a
number $p$ and learning whether $p \leq x$ or whether $p > x$. One can view this as
an online learning problem, where every round $t$ a learner guesses a value $p_t
\in [0,1]$, learns whether or not $p_t < x$, and incurs some loss $\ell(x, p_t)$
(for some loss function $\ell(\cdot, \cdot)$). The goal of the learner is to
minimize the total loss $\sum_{t=1}^{T} \ell(x,p_t)$ which can
alternatively be thought of as the learner's \textit{regret}. For example, for
the loss function $\ell(x, p_t) = |x-p_t|$, the learner can achieve total regret bounded
by a constant via the standard binary search algorithm.

In this paper, we consider a contextual, multi-dimensional generalization of
this problem which we call the \textit{contextual search problem}. Now, the
learner's goal is to learn the value of a hidden vector $v \in [0,1]^d$. Every
round, an adversary provides a context $u_t$, a unit vector in $\R^d$, to the
learner. The learner must now guess a value $p_t$, upon which they incur loss
$\ell(\langle u_t, v\rangle, p_t)$ and learn whether or not $p_t \leq \langle u_t,
v \rangle$. Geometrically, this corresponds to the adversary providing the
learner with a hyperplane; the learner may then translate the hyperplane however
they wish, and then learn which side of the hyperplane $v$ lies on. Again, the
goal of the learner is to minimize their total loss
$\sum_{t=1}^{T}\ell(\langle u_t, v \rangle, p_t)$. 

This framework captures a variety of problems in contextual decision-making.
Most notably, it captures the well-studied problem of \textit{contextual dynamic
pricing} \cite{amin2014repeated,cohen2016feature, nazerzadeh2016}.
In this problem, the learner takes on the role of a
seller of a large number of differentiated products. Every round $t$ the seller
must sell a new product with features summarized by some vector $u_t \in
[0,1]^d$. They are selling this item to a buyer with fixed values $v \in
[0,1]^d$ for the $d$ features (that is, this buyer is willing to pay up to
$\langle u, v\rangle$ for an item with feature vector $u$). The seller can set a
price $p_t$ for this item, and observes whether or not the buyer buys the item
at this price. If a sale is made, the seller receives revenue $p_t$; otherwise
the seller receives no revenue. The goal of the seller is to maximize their revenue over
a time horizon of $T$ rounds.

The dynamic pricing problem is equivalent to the contextual search problem with
loss function $\ell$ satisfying $\ell(\theta, p) = \theta - p$ if $\theta \geq
p$ and $\ell(\theta, p) = \theta$ otherwise. The one-dimensional variant of this problem was
first introduced by Kleinberg and Leighton \cite{kleinberg2003value},
who presented an $O(\log\log T)$
regret algorithm for this problem and showed that this was tight. Amin,
Rostamizadeh and Syed \cite{amin2014repeated} introduce the problem in its 
contextual, multi-dimensional form, but assume iid contexts. Cohen, Lobel, and
Paes Leme \cite{cohen2016feature} study the problem with adversarial contexts
and improve the $\tilde{O}(\sqrt{T})$-regret obtainable from general purpose contextual
bandit algorithms \cite{agarwal2014taming} to $O(d^2\log T)$-regret, based on approximating the current knowledge set
(possible values for $v$) with ellipsoids. This was later improved to $O(d\log T)$
in \cite{lobel2016multidimensional}.

In this paper we present algorithms for the contextual search problem with
improved regret bounds (in terms of their dependence on $T$). More specifically:

\begin{enumerate}
\item
For the symmetric loss function $\ell(\theta, p) = |\theta - p|$, we provide an algorithm that achieves regret $O(\poly(d))$. In contrast, the previous best-known algorithms for this problem (from the dynamic pricing literature) incur regret $O(\poly(d)\log T)$. 

\item \sloppy For the dynamic pricing problem, we provide an algorithm that achieves
  regret $O(\poly(d)\log\log T)$. This is tight up to a polynomial factor in
    $d$, and improves exponentially on the previous best known bounds of
    $O(\poly(d)\log T)$. 
\end{enumerate}

Both algorithms can be implemented efficiently in randomized polynomial time
(and achieve the above regret bounds with high probability).\\

\paragraph{Techniques from Integral Geometry}
Classical binary search involves keeping an interval of possible values
(the ``knowledge set'') and repeatedly bisecting it to decrease its length. In the
one-dimensional case length can both be used as a potential function to measure the
progress of the algorithm and as a bound for the loss. When generalizing to
higher dimensions, the knowledge set becomes a higher dimensional convex set and
the natural measure of progress (the volume) no longer directly bounds the
loss in each step.

To address this issue we use concepts from the field of
\textit{integral geometry}, most notably the notion of \textit{intrinsic
volumes}. The field of integral geometry (also known as geometric probability)
studies measures on convex subsets of Euclidean space which remain invariant
under rotations/translations of the space. 
One of the fundamental results in integral geometry is that in $d$ dimensions
there are $d+1$ essentially distinct different measures, of which surface
area and volume are two. These $d+1$ different measures are known as
\textit{intrinsic volumes}, and each corresponds to a dimension between $0$ and
$d$ (for example, surface area and volume are the $(d-1)$-dimensional and
$d$-dimensional intrinsic volumes respectively).

A central idea in our
algorithm for the symmetric loss function is to choose our guess $p_t$ so as to
divide one of the $d$ different intrinsic volumes in half. The choice of which
intrinsic volume to divide in half depends crucially on the geometry of the
current knowledge set. When the knowledge set is well-rounded and ball-like, we
can get away with simply dividing the knowledge set in half by volume. As the
knowledge set becomes thinner and more pointy, we must use lower and lower
dimensional intrinsic volumes, until finally we must divide the one-dimensional
intrinsic volume in half. By performing this division carefully, we can ensure
that the total sum of all the intrinsic volumes of our knowledge set
(appropriately normalized) decreases by at least the loss we incur each round.

Our algorithm for the dynamic pricing problem builds on top of the ideas
developed for the symmetric loss together with a new technique for charging
progress based on an isoperimetric inequality for intrinsic volumes that can be
obtained from the Alexandrov-Fenchel inequality. This new technique allows us to
combine the doubly-exponential buckets technique of Kleinberg and Leighton with
our geometric approach to the symmetric loss and obtain an $O_d(\log \log T)$
regret algorithm for the pricing loss.

One can ask whether simpler algorithms can be obtained for this setting using
only the standard notions of volume and width. We analyze simpler halving
algorithms and show that while they obtain $O_d(1)$ regret for the symmetric
loss, the dependency on the dimension $d$ is exponentially worse. While the
simple halving algorithms are defined purely in terms of standard geometric
notions, our analysis of them still requires tools from intrinsic geometry. For the pricing
loss case, we are not aware of any simpler technique just based on standard
geometric notions that can achieve $O_d(\log \log T)$ regret.

Finally, we would like to mention that to the best of our knowledge this is the
first application of intrinsic volumes to theoretical algorithm design.

\paragraph{Applications and Other Related Work} The main application of our result is
to the problem of contextual dynamic pricing. The dynamic pricing problem has been
extensively studied with different assumptions on contexts and valuation. Our
model is the same as the one in
Amin et al \cite{amin2014repeated},
Cohen et al \cite{CohenLL16} and Lobel et al \cite{lobel2016multidimensional}
who provide regret guarantees of $O(\sqrt{T})$, $O(d^2 \log T)$ and $O(d \log
T)$ respectively.
The problem was also studied with stochastic valuation and additional
structural assumptions on contexts in Javanmard and Nazerzadeh
\cite{nazerzadeh2016}, Javanmard \cite{javanmard2017perishability} and
Qiang and Bayati \cite{qiang2016dynamic}. This line of work relies on
techniques from statistic learning, such as greedy
least squares, LASSO and regularized  maximum likelihood estimators. The
guarantees obtained there also have $\log T$ dependency on the time horizon.

The contextual search problem was also considered with the loss function
$\ell(\theta, p) = \mathbf{1}\{ \abs{\theta - p} > \epsilon \}$. For this loss
function, Lobel et al \cite{lobel2016multidimensional} provide the optimal regret
guarantee of $O(d \log (1/\epsilon))$. The geometric techniques developed in
this line of work were later applied by Gillen et al \cite{gillen2018online}
in the design of online algorithms with an unknown fairness objective. Another important
application of contextual search is the problem of personalized medicine studied
by Bastani and Bayati \cite{bayati2016} in which the algorithms is presented
with patients who are described in terms of feature vectors and needs to decide
on the dosage of a certain medication. The right dosage for each patient might
depend on age, gender, medical history along with various other features. After 
prescribing a certain dosage, the algorithm only observes if the patient was underdosed 
or overdosed.

\paragraph{Paper organization}
The remainder of the paper is organized as follows. In Section \ref{sect:prelim}
we define the contextual search problem and related notions. In Section
\ref{sect:oned} we review what is known about this problem in one dimension
(where contexts are meaningless), specifically the $O(\log\log T)$ regret
algorithm of Leighton and Kleinberg for the dynamic pricing problem and the
corresponding $\Omega(\log \log T)$ lower bound. In Section \ref{sect:twod}, we
present our algorithms for the specific case where $d=2$, where the relevant
intrinsic volumes are just the area and perimeter, and where the proofs of
correctness require no more than elementary geometry (and the 2-dimensional
isoperimetric inequality). In Section \ref{sect:intrinsic}, we define intrinsic
volumes formally and introduce all relevant necessary facts. In Section
\ref{sect:multid}, we present our two main algorithms in their general form,
prove upper bounds on their regret, and argue that they can be implemented
efficiently in randomized polynomial time. In Section
\ref{sect:halving}, we consider simple halving algorithms (such as those that
always halve the width or volume of the current knowledge set) and analyze their
regret using our tools from integral geometry. Finally in Section
\ref{sect:general_loss} we
discuss how to generalize our algorithms to other loss functions.

\section{Preliminaries}\label{sect:prelim}

\subsection{Contextual Search}


We define the \emph{contextual search problem} as a game between between
a \emph{learner} and an \emph{adversary}.
The adversary begins by choosing a point $v \in [0,1]^d$. 
Then, every round for $T$ rounds, the adversary
chooses a context represented by an
unit vector $u_{t} \in \R^d$ and gives it to the learner. The
learner must then choose a value $p_{t} \in \R$, whereupon the learner accumulates
regret $\ell(\dot{u_t}{v}, p_t)$ (for some loss function $\ell(\cdot, \cdot)$) and learns whether $p_t \leq \dot{u_t}{v}$ or
$p_t \geq \dot{u_t}{v}$.
The goal of
the learner is to minimize their total regret, which is equal to the sum of their
losses over all time periods: $\Reg = \sum_t \ell(\dot{u_t}{v}, p_t)$.

We primarily consider two loss functions:\\

\noindent \textbf{Symmetric loss.} The symmetric loss measures the
absolute value between the guess and the actual dot product, i.e,
$$\ell(\theta, p) = \abs{\theta - p}.$$
Alternatively, $\ell(\theta, p)$ can be thought of as the distance between the learner's hyperplane $H_t 
:= \{x \in \R^d; \dot{u_t}{x} = p_t\}$ and the adversary's point $v$.\\

\noindent \textbf{Pricing loss.} The pricing loss corresponds to the revenue
loss by pricing an item at $p$ when the buyer's value is $\theta$. If a
price $p \leq \theta$ the product is sold with revenue $p$, so the loss
with respect to the optimal revenue $\theta$ is $\theta - p$. If the
price is $p > \theta$, the product is not sold and the revenue is zero,
generating loss $\theta$. In other words,
$$\ell(\theta, p) = \theta - p \one\{ p \leq \theta\}.$$
The pricing loss function is highly asymmetric: underpricing by $\epsilon$ can
only cause the revenue to decrease by $\epsilon$ while overpricing
by $\epsilon$ can cause the item not to be sold generating a large loss.

\subsection{Notation and framework}\label{sec:notation}

The algorithms we consider will keep track of a \textit{knowledge set} $S_t
\subseteq S_1 :=  [0,1]^d$, which will be the set of vectors
$v$ consistent with all observations so far. In step $t$ if the context is
$u_t$ and the guess is $p_t$, the algorithm will update $S_{t+1}$ to
$S_t^+(p_t; u_t)$ or $S_t^-(p_t; u_t)$ depending on the feedback obtained, where:

$$S_t^+(p_t; u_t) := \{ x \in S_t; \dot{u_t}{x} \geq p_t \} \quad \text{ and }
S_t^-(p_t; u_t) := \{ x \in S_t; \dot{u_t}{x} \leq p_t \} $$

Since $S_1$ is originally a convex set and since $S_{t+1}$ is always obtained from
$S_t$ by intersecting it with a halfspace, our knowledge set $S_t$ will remain convex for all $t$. 

Given context $u_t$ in round $t$, we let $\underline{p}_t$ and $\overline{p}_t$ be the minimum and maximum (respectively) of the dot product $\langle u_t, x\rangle$ that is consistent with $S_t$:

$$\underline{p}_t = \min_{x \in S_t} \dot{u_t}{x} \quad \text{ and } \quad
  \overline{p}_t = \max_{x \in S_t} \dot{u_t}{x} $$

Finally, given a set $S$ and an unit vector $u$ we will define the width in the
direcion $u$ as $$\wid (S; u) = \max_{x \in S} \dot{u}{x} - \min_{x \in S}
\dot{u}{x}.$$

We will consider strategies for the learner that map the current knowledge set
$S_t$ and context $u_t$ to guesses $p_t$. In Algorithm \ref{algo:framework} we
summarize our general setup.

\begin{algorithm}[h]
 \caption{Contextual search framework}
   \begin{algorithmic}[1] \label{algo:framework}
   \STATE Adversary selects $v \in S_1 = [0,1]^d$\\
 \FOR{$t = 1$ to $T$} 
 \STATE Learner receives a unit vector $u_t \in \R^d$, $\norm{u_t}=1$.
 \STATE Learner selects $p_t \in \R$ and incurs loss $\ell(\dot{u_t}{v},
 p_t)$.
 \STATE Learner receives feedback and learns the sign of $\dot{u_t}{x} - p_t$. 
 \STATE Learner updates $S_{t+1}$ to $S_t^+(p_t;u_t)$ or $S_t^-(p_t;u_t)$ accordingly.
 \ENDFOR
 	\end{algorithmic}

\end{algorithm}

Oftentimes, we will want to think of $d$ as fixed, and consider only the asymptotic dependence on $T$ of some quantity (e.g. the regret of some algorithm). We will use the notation $O_{d}(\cdot)$ and $\Omega_{d}(\cdot)$ to hide the dependency on $d$. 

\section{One dimensional case and lower bounds}\label{sect:oned}

In the one dimensional case, contexts are meaningless and 
the adversary only gets to choose the unknown parameter $v \in [0,1]$. 
Here algorithms which achieve optimal regret (up to constant factors) are known for both the symmetic loss and the pricing loss. We review them here both as a warmup for the
multi-dimensional version and as a way to obtain lower bounds for the multi-dimensional problem.

For the symmetric loss function, binary search gives constant regret. If
the learner keeps an interval $S_t$ of all the values of $v$ that are consistent
with the feedback received and in each
step guesses the midpoint, then the loss $\ell_t \leq \abs{S_t} = 2^{-t}
\abs{S_1}$. Therefore the total regret $\Reg = \sum_t \ell_t = O(1)$.\\

For the pricing loss, one reasonable algorithm is to perform $\log T$ steps
of binary search, obtain an interval containing $v$ of length $1/T$ and price at the
lower end of this interval. This algorithm gives the learner regret $O(\log T)$. Kleinberg
and Leighton \cite{kleinberg2003value} provide a surprising algorithm that exponentially improves upon this regret. Their policy biases the search towards lower prices to guarantee that if at some point the price $p_t$ is above $v$, then the length of the interval $S_t$
decreases by a large factor.

Kleinberg and Leighton's algorithm works as follows. At all rounds, they maintain a knowledge set $S_t = [a_t, a_t + \Delta_t]$. If $\Delta_t > 1/T$, they choose the price $p_t = a_t + 1/2^{2^{k_t}}$ where $k_t = \lfloor 1+\log_2 \log_2
\Delta_t^{-1} \rfloor$ (this is approximately equivalent to choosing $p_t = a_t + \Delta_t^2$). Otherwise (if $\Delta_t \leq 1/T$), they set their price equal to $a_t$. (In Appendix \ref{appendix:one-dim} we present their analysis
of this algorithm.) Moreover, they show that this bound is tight up to constant factors:

\begin{theorem}[Kleinberg and Leighton \cite{kleinberg2003value}]
The optimal regret for the contextual search problem with pricing loss in one dimension is $\Theta(\log \log T)$.
\end{theorem}

Their result implies a lower bound for the $d$-dimensional problem. If the
adversary only uses coordinate vectors $e_i = (0 \hdots 0 1 0 \hdots 0)$ as
contexts, then the problem reduces to $d$ independent instances of the one
dimensional pricing problem.

\begin{corollary} Any algorithm for the $d$-dimensional contextual search problem with pricing loss must incur $\Omega(d \log \log T)$ regret.
\end{corollary}

\section{Two dimensional case}\label{sect:twod}


We start by showing how to obtain optimal regret for both loss functions in
the two dimensional case. We highlight this special case since it is simple to
visualize and conveys the geometric intuition for the general case. Moreover,
it can be explained using only elementary plane geometry.

\subsection{Symmetric loss}

\begin{figure}
\centering
\begin{subfigure}[b]{0.30\textwidth}
\begin{tikzpicture}[scale=.9, xscale=1.4]
  \fill[blue!0!white] (-.1,-.4) rectangle (3.4,2.6);
  \draw[line width=1.5pt] (0,1) .. controls (0,1.4) and (.6, 2) .. (1,2)
              .. controls (1.4,2) and (3,1.4) .. (3,1)
              .. controls (3,0.6) and (1.4,0) .. (1,0)
              .. controls (.6,0) and (0,.6) .. (0,1);
 \node [shape=circle, fill=black,inner sep=1.5pt,label=above:$x_1$] (X1) at
 (1.5,1.89) {};
\node [shape=circle, fill=black,inner sep=1.5pt,label=below:$x_2$] (X2) at
 (1.5,2-1.89) {};
 \draw[line width=1.2pt, color=blue] (X1)--(X2);
 \node at (1.7,1) {$h$};
 \begin{scope}[line width=1.0pt]
 \begin{scope}[>=latex]
 \draw[<->] (0,2.5)--(1.5,2.5);
 \draw[<->] (3,2.5)--(1.5,2.5);
 \draw[->] (2.6,1.7) -- (3.2,1.7);  
 \end{scope}
 \draw (0,2.4)--(0,2.6);
 \draw (1.5,2.4)--(1.5,2.6);
 \draw (3, 2.4)--(3,2.6);
 \end{scope}
 \node at (1.5/2, 2.7) {$w$};
 \node at (1.5 + 1.5/2, 2.7) {$w$};
 \node at (2.8,1.9) {$u_t$};
\end{tikzpicture}
\caption{$w \geq h$ : midpoint cut}
\label{fig:2dsymm_a}
\end{subfigure}
 \begin{subfigure}[b]{0.30\textwidth}
\begin{tikzpicture}[scale=.9, xscale=1.4]
  \fill[blue!0!white] (-.1,-.4) rectangle (4,2.6);
  \draw[line width=1.5pt] (0,1) .. controls (0,1.4) and (.6, 2) .. (1,2)
              .. controls (1.4,2) and (3,1.4) .. (3,1)
              .. controls (3,0.6) and (1.4,0) .. (1,0)
              .. controls (.6,0) and (0,.6) .. (0,1);
 \node [shape=circle, fill=black,inner sep=1.5pt,label=below:$x_1$] (X1) at
 (0,1) {};
\node [shape=circle, fill=black,inner sep=1.5pt,label=below:$x_2$] (X2) at
 (3,1) {};
 \draw[line width=1.2pt, color=blue] (X1)--(X2);
 \node at (1.5,1.2) {$h$};
 \begin{scope}[line width=1.0pt]
 \begin{scope}[>=latex]
 \draw[<->] (3.2,0)--(3.2,1);
 \draw[<->] (3.2,1)--(3.2,2);
 \draw[->] (2.8,1.4) -- (2.8,2.2);  
 \end{scope}
 \draw (3.1,0)--(3.3,0);
 \draw (3.1,1)--(3.3,1);
 \draw (3.1,2)--(3.3,2);
 \end{scope}
 \node at (3.4, 1.5) {$w$};
 \node at (3.4, .5) {$w$};
 \node at (2.6, 1.7) {$u_t$};
\end{tikzpicture}
\caption{$w < h$: half area cut}
\label{fig:2dsymm_b}
\end{subfigure}
\begin{subfigure}[b]{0.30\textwidth}
\begin{tikzpicture}[scale=.9, xscale=1.4]
  \fill[blue!0!white] (-.1,-.4) rectangle (4,2.6);
\filldraw[fill=red!20!white, draw=red, line width=1.5pt] (0,1) -- (1.5,2-1.89) -- (3,1) -- (1.5,1.89) -- cycle;
\draw[line width=1.5pt] (0,1) .. controls (0,1.4) and (.6, 2) .. (1,2)
              .. controls (1.4,2) and (3,1.4) .. (3,1)
              .. controls (3,0.6) and (1.4,0) .. (1,0)
              .. controls (.6,0) and (0,.6) .. (0,1);
 \node [shape=circle, fill=black,inner sep=1.5pt,label=above:$x_1$] (X1) at
 (1.5,1.89) {};
\node [shape=circle, fill=black,inner sep=1.5pt,label=below:$x_2$] (X2) at
 (1.5,2-1.89) {};
 \node [shape=circle, fill=black,inner sep=1.5pt,label=left:$x_{\min}$] (Xmin) at (0,1) {};
 \node [shape=circle, fill=black,inner sep=1.5pt,label=right:$x_{\max}$] (Xmax) at
(3,1) {};
\end{tikzpicture}
\caption{area $\geq wh$}
\label{fig:2dsymm_proof}
\end{subfigure}

\caption{}
\end{figure}

Our general approach will be to maintain a potential function of the current knowledge set which decreases each round by an amount proportional to the loss. Since at each time $t$, the loss is bounded by the width $\wid(S_t, u_t)$ of the knowledge $S_t$ in direction $u_t$, it suffices to show that our potential function decreases each round by some amount proportional to the width of the current knowledge set.

What should we pick as our potential function? Inspired by the one-dimensional case, where one can take the potential function to be the length of the current interval, a natural candidate for the potential function is the area of the current knowledge set. Unfortunately, this does not work; if the knowledge set is long in the direction of $u_t$ and skinny in the perpendicular direction (e.g. Figure \ref{fig:2dsymm_a}), then it can have large width but arbitrarily small area.

Ultimately we want to make the width of the knowledge set as small as
possible in any given direction. This motivates a second choice of potential
function: the average width of the knowledge set, i.e., $\frac{1}{2\pi} \int_0^{2 \pi} \wid(S_t,
u_\theta) d\theta$ where
$u_\theta = (\cos \theta, \sin \theta)$. A result of Cauchy (see Section 5.5 in
\cite{klain1997introduction}) shows that the average width of a convex 2-dimensional shape is proportional to the perimeter, so this potential function can alternately be thought of as the perimeter of the knowledge set. 

Unfortunately, this too does not quite work. Now, if the set $S_t$ is thin in the direction $u_t$ and long in the perpendicular direction (e.g. Figure \ref{fig:2dsymm_b}), any cut will result in a negligible decrease in perimeter (in particular, the perimeter decreases by $O(w^2)$ instead of $\Theta(w)$). 

This motivates us to consider an
algorithm that keeps track of two potential functions: the perimeter $P_t$, and the square root of the area
$\sqrt{A_t}$. Each iteration, the algorithm will (depending on the shape of the knowledge set) choose one of these two potentials to make progress in.  If $S_t$ is long in the
$u_t$ direction, cutting it through the midpoint will allow us to decrease
the perimeter by an amount proportional to the loss incurred
(Figure \ref{fig:2dsymm_a}). If $S_t$ is thin in the
$u_t$ direction, then we can charge the loss to the square root of the area
(Figure \ref{fig:2dsymm_b}).

In Algorithm \ref{algo:symmetric2d} we describe how to compute the
guess $p_t$ from the knowledge set $S_t$ and $u_t$. Recall that the full setup
together with how knowledge sets are updated is defined in Algorithm
\ref{algo:framework}.

\begin{algorithm}[h]
  \caption{ 2D-SymmetricSearch }
  \begin{algorithmic}[1] \label{algo:symmetric2d}
 \STATE $w = \frac{1}{2}(\overline{p}_t - \underline{p}_t)$ and $p^{\mid}_t = \frac{1}{2}(\overline{p}_t + \underline{p}_t )$
 \STATE $h = $  length of the segment $S_t \cap \{x; \dot{u_t}{x}=p^{\mid}_t\}$
 \IF{$w \geq h$}
 	\STATE set $p_t = p^{\mid}_t$
 \ELSE 
 	\STATE set $p_t$ such that $\area(S_t^+) = \area(S_t^-)$
 \ENDIF
 	\end{algorithmic}
\end{algorithm}

\begin{theorem}
  The 2D-SymmetricSearch algorithm
  (Algorithm \ref{algo:symmetric2d}) has regret bounded by 
  $8 + 2\sqrt{2}$ for the symmetric loss.
\end{theorem}

\begin{proof}
We will keep track of the perimeter $P_{t}$ and the area $A_t$ of the
knowledge set $S_t$ and consider the potential function $\Phi_t = P_t +
  \sqrt{A_t} / C$, where the constant $C = (1-\sqrt{1/2})/2$.
We will show that every round this potential decreases by at least the regret we incur that round:
$$\Phi_t - \Phi_{t+1} \geq  |\dot{u_t}{v} - p_t|.$$
This implies that the total
regret is bounded by $\Reg \leq \Phi_1 = 4 + 2 / (1-\sqrt{1/2}) = 8 + 2\sqrt{2} $. We will write $\ell_t$ as shorthand for the loss $\ell(\dot{u_t}{v}, p_t) = |\dot{u_t}{v} - p_t|$ at time $t$. 

We first note that both $P_{t}$ and $A_{t}$ are decreasing in $t$. This
follows from the fact that $S_{t+1}$ is a convex subset of $S_{t}$. We will
show that when $w \geq h$, $P_{t}$ decreases by at least $\ell_{t}$,
whereas when $w < h$, $\sqrt{A_{t}}$ decreases by at least $\ell_{t}$.\\

\noindent \emph{Case $w \geq h$.} In this case, $p_t = p_t^{\mid}$.
We claim here that $P_{t} - P_{t+1} \geq w$. To see this, let $x_1$
and $x_2$ be the two endpoints of the line segment  $S_t \cap \{x;
\dot{u_t}{x}=p^{\mid}_t\}$ (so that $h = \norm{x_1 - x_2}$).
Without loss of generality, assume $S_{t+1} = S_{t}^-$ (the other case is
analogous).

Note that the boundary of $S_{t+1}$ is the same as the boundary of
$S_{t}$, with the exception that the segment of the boundary of $S_{t}$ in the
half-space $\{x;\dot{ u_{t}}{ x } \geq p_t^{\mid}\}$ has been replaced with
the line segment $\overline{x_1x_2}$. The part of boundary of $S_t$ in the
halfspace $\{\dot{u_{t}}{x} \geq p_{mid}\}$,
reach some point on the line $\langle u_{t}, x \rangle = \overline{p}_t$, and
return to $x_2$. Since $\overline{p}_t - p^{\mid}_t = w$, any such
path must have length at least $2w$. From the fact that $w \geq h$, it follows that:

$$ P_{t} - P_{t+1} \geq 2w - h 
\geq  2w - w 
\geq  w \geq \ell_t.$$

\noindent \emph{Case $w < h$.} We set $p_t$ such that $A_{t+1} = A_{t}/2$, so
$$\Phi_t - \Phi_{t+1} \geq (\sqrt{A_{t}} - \sqrt{A_{t+1}})/C \geq  \sqrt{A_{t}} (1-\sqrt{1/2})
/ C =  2 \sqrt{A_{t}}.$$ If we argue
that $2 \sqrt{A_{t}} \geq 2 w \geq \ell_t$ we are done. 
To show this, define $x_1$ and $x_2$ as
before, and let $x_{\min}$ be a point in $S_{t}$ that satisfies $\dot{ u_{t}}{
x_{min}} = \underline{p}_t$, and likewise let $x_{\max}$ be a point in $S_{t}$
that satisfies $\dot{ u_{t} }{ x_{max}}  = \overline{p}_t$. Since $S_{t}$ is
convex, it contains the two triangles with endpoints $(x_1, x_2, x_{max})$ and
$(x_1, x_2, x_{min})$, see Figure \ref{fig:2dsymm_proof}. 
But these two triangles are disjoint and each have area
at least $\frac{1}{2}wh$, so $A_{t} \geq wh \geq w^2$, since $h > w$. It follows
that $\sqrt{A_{t}} \geq w$.

\end{proof}

\subsection{Pricing loss}

To minimize regret for pricing loss, we want to somehow combine our above insight of looking at both the area and the perimeter with Leighton and Kleinberg's bucketing procedure for the 1D case. At first we might try to do this bucketing procedure just with the area. 

That is, if the area of the current knowledge set belongs to the interval $[\Delta^2, \Delta]$, choose a price that carves off a subset of total area $\Delta^2$. Now, if you overprice, you incur $O(1)$ regret, but the area of your knowledge set shrinks to $\Delta^2$ (and belongs to a new ``bucket''). On the other hand, if you underprice, your area decreases by at most $\Delta$, so you can underprice at most $\Delta^{-1}$ times. If the regret you incurred this round was at most $O(\Delta)$, then this means that you would incur at most $O(1)$ total regret underpricing while your area belongs to this interval.

This would be true if the regret per round was at most the area of the knowledge set. Unfortunately, as noted earlier, this is not true; the regret per round scales with the width of the knowledge set, not the area, and you can have knowledge sets with large width and small area. The trick, as before, is to look at both area and perimeter, and argue that at each step the bucketization argument for at least one of these quantities holds.

More specifically, let $P_t$ again be the perimeter of the knowledge set at time $t$, $A_t$ be the area of the knowledge set at time $t$, and let $A'_t = 2\sqrt{\pi A_t}$ be a normalization of $A_t$. Note that by the isoperimetric inequality for dimensions (which says that of all shapes with a given area, the circle has the least perimeter), we have that $P_t \geq A'_t$. 

Let $\ell_k = \exp(-1.5^k)$. Our buckets will be the intervals $(\ell_{k+1}, \ell_k]$. We will define the function $\bkt(x) = k$ if $x \in (\ell_{k+1}, \ell_k]$. Our algorithm is described in Algorithm \ref{algo:pricing2d}. The analysis follows.

\begin{algorithm}[h]
  \caption{ 2D-PricingSearch }

   \begin{algorithmic}[1] \label{algo:pricing2d}
\STATE $w = (\overline{p}_t - \underline{p}_t)$
\STATE $h_{max} = $ maximum length of a segment of the form $S_t \cap \{x; \dot{u_t}{x}=p\}$ 
\IF{$w < 1/T$} 
	\STATE choose $p_t = \underline{p}_t$. 
\ELSIF{$\bkt(A'_t) = \bkt(P_t)$} 
  \STATE set $p_t$ such that $\area(S_t^{-}) = \ell^2_{\bkt(A'_t)+1}/4\pi$. 
\ELSIF{$\bkt(A'_t) > \bkt(P_t)$ and $w < h_{max}$}
  \STATE set $p_t$ such that $\area(S_t^{-}) = \ell^2_{\bkt(A'_t)+1}/4\pi$. 
\ELSIF{$\bkt(A'_t) > \bkt(P_t)$ and $w \geq h_{max}$} 
  \STATE set $p_t$ such that $\perim(S_t) - \perim(S_t^{+}) = \frac{1}{2}\ell_{\bkt(P_t) + 1}$.
\ENDIF
\end{algorithmic}

\end{algorithm}

\begin{theorem}
  The 2D-PricingSearch algorithm
  (Algorithm \ref{algo:pricing2d}) has regret bounded by $O(\log\log T)$ for the pricing loss.
\end{theorem}

\begin{proof}
We will divide the behavior of the algorithm into four cases (depending on which branch of the if statement in Algorithm \ref{algo:pricing2d} is taken), and argue the total regret sustained in each case is at most $O(\log\log T)$.

\begin{itemize}
\item
\textit{Case 1: $w \leq 1/T$.} Whenever this happens, we pick the minimum possible price in our convex set, so we definitely underprice and sustain regret at most $1/T$. The total regret sustained in this case over $T$ rounds is therefore at most $1$.

\item
\textit{Case 2: $\bkt(A'_t) = \bkt(P_t)$.} Let $\bkt(P_t) = \bkt(A'_t) = k$. Note that $k \leq O(\log\log T)$, or else we would be in Case 1 (if $P_t < 1/T$, then $w < 1/T$). Therefore, fix $k$; we will show that the total regret we incur for this value of $k$ is at most $O(1)$. Summing over all $O(\log\log T)$ values of $k$ gives us our upper bound.

If we overprice and do not make a sale, this contributes regret at most $1$, but then $A'_{t+1} = \ell_{k+1}$ and we leave this bucket. If we underprice, our regret is at most $w \leq P_t \leq \ell_k$, and the decrease in area $A_{t+1} - A_{t} = \ell_{k+1}^2/4\pi$. It follows that we can underprice at most $4\pi\ell_{k}^2/\ell_{k+1}^2$ before leaving this bucket, and therefore we incur total regret at most

$$4\pi\frac{\ell_k^2}{\ell_{k+1}^2} \ell_{k} = 4\pi\exp(-3\cdot 1.5^k + 2\cdot 1.5^{k+1}) = 4\pi = O(1).$$

\item
\textit{Case 3: $\bkt(A'_t) > \bkt(P_t)$ and $w < h_{max}$.} Let $\bkt(A'_t) = r$ and $\bkt(P_t) = k$. As in case 2, we will fix $r$ and argue that total regret we incur for this $r$ is at most $O(1)$. As before, if we overprice, $A'_{t+1}$ becomes $\ell_{r+1}$, so we incur total regret at most $O(1)$ from overpricing. 

Now, note that since $w < h_{max}$, $S_t$ contains two disjoint triangles with base $h_{max}$ and combined height $w$ (see Figure \ref{fig:2dsymm_proof}), so $A_t \geq wh_{max} > w^2$, and therefore $w < \sqrt{A_t} = A'_t/2\sqrt{\pi}$. Therefore, if we underprice, we incur regret at most $A'_t/2\sqrt{\pi} \leq \ell_{r}$. As before, the area decreases by at least $A_{t+1} - A_t = \ell_{r+1}^2/4\pi$ if we underprice, so we underprice at most $4\pi\ell_{k}^2/\ell_{k+1}^2$ before leaving this bucket, and therefore we incur total regret at most

$$4\pi\frac{\ell_r^2}{\ell_{r+1}^2} \ell_{r} = 4\pi\exp(-3\cdot 1.5^r + 2\cdot 1.5^{r+1}) = 4\pi = O(1).$$

\item
\textit{Case 4: $\bkt(A'_t) > \bkt(P_t)$ and $w \geq h_{max}$.} Let $\bkt(A'_t) = r$ and $\bkt(P_t) = k$, and fix $k$. Now $A_{t} \geq wh_{max} \geq h_{max}^2$, so $h_{max} < A'_t/2\sqrt{\pi} \leq \ell_{r}/2\sqrt{\pi}$. 

First, note that it is in fact possible to set $p_t$ so that $\perim(S_t) - \perim(S_t^{+}) = \frac{1}{2}\ell_{k+1}$, since the total perimeter is at least $\ell_{k+1}$, so this corresponds to just cutting off a chunk ($S_t^{+}$) of $S_t$ of perimeter $\perim(S_t) - \frac{1}{2}\ell_{k+1}$ (which is possible since this is nonnegative and less than $\perim(S_t)$).

    If we overprice, then the perimeter of the new region ($S_{t}^{-}$) is equal to $\perim(S_t) - \perim(S_{t}^{+}) = \ell_{k+1}/2$, plus the length of the segment formed by the intersection of $S_t$ with $\{x; \langle u_t, x\rangle = p_t\rangle$. The length of this segment is at most $h_{max}$, so the perimeter is at most $\ell_{k+1}/2 + \ell_{r}/2\sqrt{\pi} \leq \ell_{k+1}/2 + \ell_{k+1}/2\sqrt{\pi} \leq \ell_{k+1}$. This means we can overprice at most once before the perimeter changes buckets, and we thus incur at most $O(1)$ regret due to overpricing.

If we underprice, then we incur regret at most $w \leq P_t \leq \ell_k$, and the perimeter of the new region decreases by at least $\ell_{k+1}/2 - \ell_{r}/2\sqrt{\pi} \geq \ell_{k+1}/5$. This means we can underprice at most $5\ell_{k}/\ell_{k+1}$ times before we switch buckets, and we incur total regret at most

$$5\frac{\ell_{k}}{\ell_{k+1}}\ell_k = 4\exp(-2\cdot 1.5^{r} + 1.5^{r+1}) \leq 4 = O(1).$$

\end{itemize}

\end{proof}

\section{Interlude: Intrinsic Volumes}\label{sect:intrinsic}

The main idea in the two dimensional case was to balance between making progress
in a two-dimensional measure of the knowledge set (the area) and in a
one-dimensional measure (the perimeter). To generalize this idea to
higher dimensions we will keep track of $d$ potential functions, each corresponding to a
$j$-dimensional measure of the knowledge set for each $j$ in $\{1, 2, \dots d\}$. 

Luckily for us, one
of the central objects of study in \emph{integral geometry} (also known as \textit{geometric probability}) corresponds
exactly to a $j$-dimensional measure of a $d$-dimensional object. Many readers are undoubtedly familiar with two of these measures, namely volume (the $d$-dimensional measure) and surface area (the $(d-1)$-dimensional area) but it is less clear how to define the $1$-dimensional measure of a three-dimensional convex set (indeed, Shanuel
\cite{schanuel1986length} jokingly calls his lecture notes on the topic ``What is
the length of a potato?"). These measures are known as \emph{intrinsic volumes} and they 
match our intuition for how a $j$-dimensional measure of a $d$-dimensional set should behave (in particular, reducing to the regular $j$-dimensional volume as the set approaches a $j$-dimensional object).

We now present a formal definition of intrinsic volumes and summarize their most
important properties. We refer to the excellent book by Klain and Rota
\cite{klain1997introduction} for a comprehensive introduction to integral geometry.

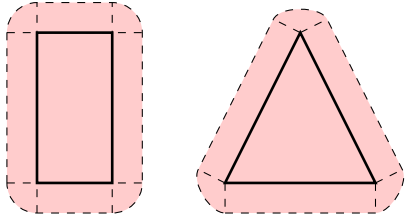
\begin{figure}
\centering
\begin{tikzpicture}
  \draw[dashed, fill=red!20!white] (-.4,0) -- (-.4,2) .. controls (-.4,2.22) and (-0.22,2.4) .. (0,2.4) --
  (1,2.4) .. controls  (1.22,2.4) and (1.4,2.22)  .. (1.4,2) -- (1.4,0) ..
  controls    (1.4,-.22)  and (1.22,-.4)  .. (1,-.4) -- (0,-.4) .. controls
  (-.22,-.4) and (-.4,-.22)  .. cycle;

  \draw[line width = 1] (0,0) -- (1,0) -- (1,2) -- (0,2) -- cycle;
  \draw[dashed] (-.4, 0) -- (1.4,0);
  \draw[dashed] (-.4, 2) -- (1.4,2);
  \draw[dashed] (0, -.4) -- (0 ,2.4);
  \draw[dashed] (1, -.4) -- (1 ,2.4);

  \begin{scope}[shift={(2.5,0 )}]
  \draw[dashed, fill=red!20!white] (0,-.4) -- (2,-.4) .. controls (2+.2,-.4) and
    (2.357 + .2*.447 ,.178 - .2*.894) .. (2.357,.178) --
  (2.357-1,.178+2) .. controls  (2.357-1 - .2*.447,.178+2 + .2*.894)  and  
    (2-2.357+1+.2*.447,.178+2 + .2*.894)
 .. (2-2.357+1,.178+2)
    -- (2-2.357,.178) ..
  controls    (2-2.357 - .2*.447 ,.178 - .2*.894)  and (-.2,-.4)   .. cycle;

    \draw[dashed] (0,0) -- (0,-.4);
    \draw[dashed] (2,0) -- (2,-.4);
    \draw[dashed] (0,0) -- ( -.4*0.894,  .4*0.447 );
    \draw[dashed] (2,0) -- ( 2+.4*0.894,  .4*0.447 );
    \draw[dashed] (1,2) -- ( 1-.4*0.894, 2+ .4*0.447 );
    \draw[dashed] (1,2) -- ( 1+.4*0.894, 2+ .4*0.447 );


    \draw[line width = 1] (0,0) -- (2,0) -- (1,2) -- cycle;
  \end{scope}
\end{tikzpicture}
  \caption{Steiner's formula for  2D: $\sf{Area}(K + \epsilon B) =
  \sf{Area}(K) + \sf{Perimeter}(K) \cdot \epsilon +  \pi \epsilon^2 $}
\label{fig:steiner}
\end{figure}

Intrinsic volumes can be defined as the coefficients that arise in Steiner's
formula for the the volume of the (Minkowski) sum of a convex
set $K \subseteq \R^d$ and an unit ball $B$. Steiner
\cite{schanuel1986length} shows that the
$\Vol(K+\eps B)$ is a polynomial in $\epsilon$ and the intrinsic volumes $V_j(K)$ are the
(normalized) coefficients of this polynomial:
\begin{equation}\label{eqn:steineralt}
\Vol(K+\eps B) = \sum_{j=0}^{d}\kappa_{d-j}V_j(K)\eps^{d-j}
\end{equation}
where $\kappa_{d-j}$ is the volume of the $(d-j)$-dimensional unit ball. An
useful exercise to get intuition about intrinsic volumes is to directly compute
the intrinsic volumes of the parallelotope $[0,a_1] \times [0,a_2] \times
[0,a_d]$. It is easy to check for $d=2$ and $3$ (see Figure \ref{fig:steiner}) that
$V_d = a_1 a_2 \hdots a_d$, $V_1 = a_1 + a_2 + \hdots + a_d$ and $V_0 = 1$.
More generally $V_j$
correspnds to the symmetric polynomial of degree $j$:
$V_j = \sum_{S \subseteq [d]; \abs{S} = j} a_S$ where $a_S = \prod_{s \in S} a_s$. 
In particular for $[0,1]^d$ the $j$-th intrinsic volume is
$V_j = {d \choose j}$.

\begin{definition}[Valuations] Let $\Conv_d$ be the class of compact convex bodies in
  $\R^d$. A \emph{valuation} is a map $\nu:\Conv_d \rightarrow \R$ such that
  $\nu(\emptyset) = 0$ and for every $S_1, S_2 \in \Conv_d$ satisfying $S_1\cup S_2 \in \Conv_d$ it holds that
  $$\nu(S_1 \cup S_2) +\nu(S_1 \cap S_2) = \nu(S_1) + \nu(S_2).$$
  A valuation is said to be \emph{monotone} if $\nu(S) \leq \nu(S')$ whenever $S
  \subseteq S'$. A valuation is said to be \emph{non-negative} is $\nu(S) \geq 0$.
  Finally, a valuation is \emph{rigid} if $\nu(S) = \nu(T(S))$ for every rigid motion
  (i.e. rotations and translations) $T$ of $\R^d$.
\end{definition}

To define what it means for a valuation to be continuous, we need a notion of
distance betweeen convex sets. We define the Hausdorff distance $\delta(K,L)$
between two sets $K,L \in \Conv_d$ to be the  the minimum $\epsilon$ such that $K +
\epsilon B \subseteq L$ and $L + \epsilon B \subseteq K$ where $B$ is the unit
ball. This notion of distance allows us to define limits: a sequence
$K_t \in \Conv_d$ converges to $K \in \Conv_d$ (we write this as $K_t
\rightarrow K$) if $\delta(K_t, K) \rightarrow 0$. Continuity can now be defined
in the natural way.

\begin{definition}[Continuity]\label{def:continuous}
  A valuation function $\nu$ is \emph{continuous} if
  whenever $K_t \rightarrow K$ then $\nu(K_t) \rightarrow \nu(K)$.
\end{definition}

\begin{theorem}\label{thm:valuation}
The intrinsic volumes are non-negative monotone continuous rigid valuations.
\end{theorem}

In fact the intrinsic volumes are quite special since they form a basis for the
set of all valuations with this property. This constitutes the fundamental
result of the field of integral geometry:

\begin{theorem}[Hadwiger]\label{thm:hadwiger}
  If $\nu$ is a continuous rigid valuation of $\Conv_d$,
  then there are constants $c_0, \hdots, c_d$ such that $\nu = \sum_{i=0}^d c_i
  V_i$.
\end{theorem}

Next we describe a few important properties of intrinsic valuations that will be
useful in the analysis of our algorithms:

\begin{theorem}[Homogeneity] The map $V_j$ is $j$-homogenous,
  i.e., $V_j(\alpha K) = \alpha^j V_j(K)$ for any $\alpha \in \R_{\geq 0}$.
\end{theorem}

\begin{theorem}[Ambient independence]\label{thm:ambient}
  Intrinsic volumes are independent of the
  ambient space, i.e, if $K \in \Conv_d$ and $K'$ is a copy of $K$ embedded in a
  larger dimensional space
$$K' = T(\{(x, 0_k); x \in K, 0_k \in \R^k \} \in
  \Conv_{d+k}$$ for a rigid transformation $T$, then for any $j \leq d$, we have
$V_j(K) = V_j(K')$.
\end{theorem}

We now provide an inequality between intrinsic volumes which we will use later to derive an isoperimetric inequality for intrinsic volumes. The following inequality
is a consequence of the Alexandrov-Fenchel inequality due to McMullen
\cite{mcmullen1991inequalities}.

\begin{theorem}[Inequality on intrinsic volumes]\label{thm:af}
If $S \in \Conv_d$ and any $i \geq 1$ then 
$$V_i(S)^2 \geq \frac{i+1}{i} V_{i-1}(S) V_{i+1}(S).$$
\end{theorem}

One beautiful consequence of Hadwiger's theorem is a probabilistic
interpretation of intrinsic volumes as the expected volume of the projection of
a set onto a random subspace. To make this precise, define the
Grassmannian $\Gr(d, k)$ as the collection of all $k$-dimensional
linear subspaces of $\R^d$. The \emph{Haar measure} on the Grassmannian is the
unique probability measure on $\Gr(d,k)$ that is invariant under rotations in $\R^d$ (i.e., $SO(\R^d)$).

\begin{theorem}[Random Projections]\label{thm:random_projections}
  For any $K \in \Conv_d$, the $j$-th intrinsic volume $$V_j(K) = \E_{H \sim
  \Gr(d,k)}[\Vol(\pi_H(K))]$$ where $H \sim \Gr(d,k)$ is a $k$-dimensional subspace
  $H$ sampled according to the Haar measure, $\pi_H$ is the projection on $H$
  and $\Vol(\pi_H(K))$ is the usual ($k$-dimensional) volume on $H$.
\end{theorem}

A remark on notation: we use $\Vol$ to denote the standard notion of volume and
$V_j$ to denote intrinsic volumes. When analyzing an object in a $d$-dimensional
(sub)space, then $\Vol = V_d$.

\section{Higher dimensions}\label{sect:multid}

In this section, we generalize our algorithms from Section \ref{sect:twod} from the two-dimensional case to the general multi-dimensional case.

Both results require as a central component lower bounds on the intrinsic volumes of high dimensional cones. These bounds relate the intrinsic volume of a cone to the product of the cone's height and the intrinsic volume of the cone's base (a sort of  ``Fubini's theorem'' for intrinsic volumes).

More formally, a \textit{cone} $S$ in $\R^{d+1}$ is the convex hull of a $d$-dimensional convex set
$K$ and a point $p \in \R^{d+1}$. If the distance from $p$ to the affine subspace
containing $K$ is $h$,
we say the cone has \textit{height} $h$ and \textit{base} $K$. The lemma we require is the following.

\begin{lemma}[Cone Lemma]\label{lem:cone}
Let $K$ be a convex set in $\R^{d}$, and let $S$ be a cone in $\R^{d+1}$ with base $K$ and height $h$. Then, for all $0 \leq j \leq d$,
$$V_{j+1}(S) \geq \frac{1}{j+1}h V_{j}(K).$$
\end{lemma}

In the two-dimensional case, this lemma manifests itself when we use the fact that the perimeter of a convex set with height $h$ is at least $h$. We note that when $j = d$, Lemma \ref{lem:cone} holds with equality and is a
simple exercise in elementary calculus. On the other hand, when $0 \leq j < d$, there is no
straightforward formula for the $(j+1)$-th intrinsic volume of a set in terms of the $j$-th intrinsic volume of its cross sections. 

We begin by taking the Cone Lemma as true, and discuss how to use it to generalize our
contextual search algorithms to higher dimensions in Sections \ref{sec:sym_loss} (for symmetric loss) and \ref{sec:price_loss} (for pricing loss). We then prove the Cone Lemma in Section~\ref{sec:cone_lemma}. Finally, in Section \ref{sec:efficient}, we argue that both algorithms can be implmented efficiently.

\subsection{Symmetric loss}\label{sec:sym_loss}

In this section we present a $O_{d}(1)$ regret algorithm for the contextual search problem with symmetric loss in $d$ dimensions. The algorithm, which we call SymmetricSearch, is presented in Algorithm \ref{algo:symmetricdd}.

Recall that in two dimensions, we always managed to choose $p_t$ so that the loss from that round is bounded by the decrease in either the perimeter or the square root of the area. The main idea of Algorithm \ref{algo:symmetricdd} is to similarly choose $p_t$ such that the loss is bounded by the decrease in one of the intrinsic volumes, appropriately normalized. As before, if the width is large enough, we bound the loss by the 
decrease in the average width (i.e. the one-dimensional intrinsic volume $V_1(S)$). As the width gets smaller, we charge the loss to progressively higher-dimensional intrinsic volumes.

Constants $c_0$ through $c_{d-1}$ in Algorithm \ref{algo:symmetricdd} are defined so that $c_{0} =
1$ and $c_{i}/c_{i-1} = \frac{1}{2i}$. In other words, $c_i =
\frac{1}{2^{i-1}i!}$. Constant $c_0$ is only used in the analysis.

\begin{algorithm}[h]
 \caption{SymmetricSearch}
 \begin{algorithmic}[1]  \label{algo:symmetricdd}
 \STATE $w = \frac{1}{2}\wid(S_t; u_t)$
 \FOR{$i=1$ to $d$}
  \STATE define $p_i \in \R$ such that $V_i(S_t^+(p_i; u_t)) =
 V_i(S_t^-(p_i; u_t))$.
  \STATE define $K_i = \{x \in S_t; \dot{u_t}{x} = p_i \}$.
  \STATE define $L_i = (V_i(K_i) / c_i)^{1/i}$ (set $L_0 = \infty$).
 \ENDFOR
  \STATE find $j$ such that $L_{j-1} \geq w \geq L_j$\\
  \STATE set $p_t = p_j$. 
\end{algorithmic}

\end{algorithm}

We first argue that this algorithm is well-defined:

\begin{lemma} SymmetricSearch (Algorithm \ref{algo:symmetricdd}) is well-defined, i.e., there is
  always a choice of $p_i$ and $j$ that satisfies the required properties.
\end{lemma}

\begin{proof}
We begin by arguing that there exists a $p_i$ such that $V_i(S_t^+(p_i; u_t)) = V_i(S_t^-(p_i;
u_t))$. To see this, note that the functions $\phi^+(x) =V_i(S_t^+(x; u_t))$ and
$\phi^-(x) = V_i(S_t^-(x; u_t))$ are continuous on $[\underline{p}_t,
\overline{p}_t]$ since the intrinsic volumes are continuous with respect to
  Hausdorff distance (Definition \ref{def:continuous} and Theorem
  \ref{thm:valuation}).
Moreover, since intrinsic volumes are monotone (Theorem \ref{thm:valuation}), $\phi^{+}(x)$ is decreasing and $\phi^-$ is increasing on this interval. Finally, since $\phi^{+}(\underline{p}_t) = \phi^-(\overline{p}_t)$, it follows from the Intermediate Value Theorem that there exists a $p_i$ where $\phi^+(p_i) = \phi^-(p_i)$, as desired.

To see that there exists a $j$ such that $L_{j-1}
\geq w \geq L_j$, note that $L_d = 0$ since $K_d$ is in a $d-1$-dimensional
hyperplane, so the segments $[L_{j}, L_{j-1})$ for $L_{j} < L_{j-1}$ cover the
entire $[0,\infty)$. It follows that one such interval must contain $w$.
\end{proof}

Before we proceed to the regret bound, we will show the following two lemmas. The first lemma shows that if we pick $j$ in this manner, then $V_{j}(S_t)^{1/j}$ will be at least $\Omega(w)$. 

\begin{lemma}\label{lemma:symm1}
  $V_j(S_t) \geq \frac{1}{j} c_{j-1} w^j$.
\end{lemma}

\begin{proof}
  For $j=1$, note that $S_t$ contains a segment of length $2w$, so $V_1(S_t)$
  is at most the $1$-dimensional intrinsic volume of that segment, which is
  exactly $2w$ (Theorem \ref{thm:ambient}).

  For $j > 1$, we known that $S_t$ contains a cone with base $K_{j-1}$ and
  height $w$ (since the width of $S_{t}$ is $2w$, there is a point at least
  distance $w$ from the plane $H_{j-1} = \{ x; \dot{u_t}{x} = p_{j-1} \}$).
  By Theorem \ref{lem:cone} and the fact that $V_j$ is monotone, this implies
  that:
  $$V_{j}(S_t) \geq \frac{1}{j}w V_{j-1}(K_{j-1}).$$
  Since $w < L_{j-1} = (V_{j-1}(K_{j-1}) / c_{j-1})^{1/{j-1}}$, we have that
  $V_{j-1}(K_{j-1}) \geq c_{j-1} w^{j-1}$. Substituting this into the previous
  expression, we obtain the desired result.
\end{proof}

The second lemma shows that if we pick $j$ in this manner, then $V_{j}(S_{t+1})/V_{j}(S_t)$ is bounded above by a constant strictly less than $1$.

\begin{lemma}\label{lemma:symm2}
  $V_j(S_{t+1}) \leq \frac{3}{4} V_j(S_t)$.
\end{lemma}

\begin{proof}
  The set $S_{t+1}$ is equal to either $S^+ = S^+_t(p_j;
  u_t)$ or $S^- = S^-_t(p_j;u_t)$. Our choice of $p_j$ is such that $V_j(S^-) =
  V_j(S^+)$. Therefore:
  $$2 V_j(S_{t+1}) = V_j(S^+) + V_j(S^-) = V_j(S^- \cap S^+) + V_j(S^- \cup S^+)
  = V_j(K_j) + V_j(S_t) $$
  To bound $V_j(K_j)$ in terms of $V_j(S_t)$ we observe
  that $w \geq L_j = (V_j(K_j) / c_j)^{1/j}$ so $V_j(K_j) \leq c_j w^j$.
  Plugging the previous lemma we get $V_j(K_j) \leq j
  \frac{ c_{j} }{ c_{j-1} } V_j(S_t) = \frac{1}{2} V_j(S_t)$ by the choice of
  constants. Substituting this inequality into the above equation gives us
  the desired result.
\end{proof}

Together, these lemmas let us argue that each round, the sum of the normalized intrinsic volumes $V_{i}(S_t)^{1/i}$ decreases by at least $\Omega(w)$ (and hence the total regret is constant).

\begin{theorem}\label{thm:symsearch}
  The SymmetricSearch algorithm
  (Algorithm \ref{algo:symmetricdd}) has regret bounded by 
  $O(d^4)$ for the symmetric loss.
\end{theorem}

\begin{proof}
  We will show that for the potential function $\Phi_t = \sum_{i=1}^d
  i^2 V_i(S_t)^{1/i}$ we can always charge the loss to the decrease in potential,
  i.e., $\Phi_t - \Phi_{t+1} \geq \Omega(w) \geq \Omega(\ell_t)$ and therefore,
  $\Reg \leq \sum_{t=1}^\infty \ell_t \leq O(\Phi_1)$. The initial potential is
  $$\Phi_1 = \sum_{i=1}^d i^2 V_i([0,1]^d)^{1/i} = \sum_{i=1}^d i^2 {d \choose
  i}^{1/i} \leq \sum_{i=1}^d i^2 O(d) = O(d^4)$$

  Since $V_j(S_t) \geq V_j(S_{t+1})$ by monotonicity, we can bound the potential
  change by $\Phi_t - \Phi_{t+1} \geq 
  j^2 [ V_j(S_t)^{1/j} - V_j(S_{t+1})^{1/j} ]$. We now show that this
  last term is $\Omega(w)$:
  
  \begin{eqnarray*}
  j^2 [ V_j(S_t)^{1/j} - V_j(S_{t+1})^{1/j} ] &\geq& j^2 \left(1 -\left(\frac{3}{4}\right)^{1/j}\right)V_{j}(S_t)^{1/j} \\
  &\geq & j^2 \left(1 -\left(\frac{3}{4}\right)^{1/j}\right)  \left( \frac{c_{j-1}}{j} \right)^{1/j} w  \\
  &\geq & j^2 \left(1 - (1 - \log(4/3)/j)\right)\left(\frac{1}{2^{j-2}j!}\right)^{1/j} w\\
  &\geq & j^2 \left(\frac{\log(4/3)}{j}\right)\left(\frac{e}{2j}\right)w \\
  &\geq & \Omega(w).
  \end{eqnarray*}
  Here the first inequality follows from Lemma \ref{lemma:symm2} and the
  second from Lemma \ref{lemma:symm1}. 
\end{proof}

\subsection{Pricing loss}\label{sec:price_loss}

In the 2-dimensional version of the dynamic pricing problem, we decomposed the range
of each potential into $O(\log \log T)$ buckets and used the isoperimetric
inequality $\sqrt{4 \pi A_t} \leq  P_t$ to argue that (when suitably
normalized), the area always belonged to a higher bucket than the perimeter. To apply the same idea here,
we will apply our inequality on intrinsic volumes (Theorem
\ref{thm:af}) to obtain an isoperimetric inequality for
intrinsic volumes:

\begin{lemma}[Isoperimetric inequality]\label{lem:iso}
For any $S \in \Conv_d$ and any $i \geq 1$ it holds that
$$(i! V_i(S))^{1/i} \geq ((i+1)!V_{i+1}(S))^{1/(i+1)}.$$
\end{lemma}
\begin{proof}
We proceed by induction. For $i = 1$, note that Theorem \ref{thm:af} gives us that $V_{1}(S)^2 \geq 2V_{0}(S)V_{2}(S)$. Since $V_{0}(S)$ equals $1$ for any convex set $S$, this reduces to $V_{1}(S) \geq \sqrt{2!V_{2}(S)}$. 

Now assume via the inductive hypothesis that we have proven the claim for all $j \leq i$. From Theorem \ref{thm:af} we have that

\begin{eqnarray*}
V_{i}(S)^2 &\geq & \frac{i+1}{i}V_{i-1}(S)V_{i+1}(S) = \frac{i+1}{i!} ((i-1)!V_{i-1}(S))V_{i+1}(S) \\
&\geq & \frac{i+1}{i!} (i!V_{i}(S))^{(i-1)/i}V_{i+1}(S) = \frac{1}{i!^{(i+1)/i}} V_{i}(S)^{(i-1)/i}(i+1)!V_{i+1}(S).
\end{eqnarray*}

Rearranging, this reduces to $(i!V_{i}(S))^{(i+1)/i} \geq (i+1)!V_{i+1}(S)$, and therefore $(i!V_{i}(S))^{1/i} \geq ((i+1)!V_{i+1}(S))^{1/(i+1)}$.
\end{proof}

Inspired by the isoperimetric inequality we will keep track of the following
``potentials'' (these vary with $t$, but we will omit the subscript for notational convenience):
$$\varphi_{i} = (i! V_i(S_t))^{1/i}  $$
Since $S_1 = [0,1]^d$, their initial values will be given by
$\varphi_{i} = (i! \cdot {d \choose i})^{1/i} < di
\leq d^2$. Since those quantities are monotone non-increasing, they will be in
the interval $[0,d^2)$. We will divide this interval in ranges of
doubly-exponentially decreasing length (as in one and two dimensions). The
ranges will be $(\ell_{k+1}, \ell_k]$ where
  $$\ell_k = d^2 \exp(-\alpha^k) \quad \text{for} \quad \alpha = 1+1/d$$

To keep track of which range each of our potentials $\phi_i$ belongs to, define $k_{i}$ so that $\varphi_{i} \in (\ell_{k_i+1}, \ell_{k_i}]$. By the
  isoperimetric inequality we know that:
  $$\varphi_1 \geq \varphi_2 \geq \hdots \geq \varphi_d \qquad k_1 \leq k_2 \leq \hdots
  \leq k_d$$

Recall that in the 2-dimensional case, whenever the perimeter and the area were
in the same range, we chose to make progress in the area. To extend
this idea to higher dimensions, whenever many $\phi_i$ belong to the same
range and we decide to make progress on that range, we will always choose the
largest such $\varphi_i$:

$$M(i) = \max \{j ; k_i = k_j \}$$

The complete method is summarized in Algorithm \ref{algo:pricingdd}. As before, constants $c_0$ through $c_{d-1}$ in Algorithm \ref{algo:pricingdd} are defined so that $c_{0} =1$ and $c_{i}/c_{i-1} = \frac{1}{2i}$. In other words, $c_i =\frac{1}{2^{i-1}i!}$.

We begin by arguing that our algorithm is well-defined. We ask the reader to recall the notation $\underline p_t = \min_{x \in S_t} \dot{u_t}{x}$ and $\overline p_t = \max_{x
\in S_t} \dot{u_t}{x}$.

\begin{algorithm}[h]
 \caption{PricingSearch}
 
  \begin{algorithmic}[1]  \label{algo:pricingdd}
 \STATE $w = \frac{1}{2}\wid(S_t; u_t)$
 \FOR{$i = 1$ to $d$}
 	\STATE let $\varphi_{i} = (i! \cdot V_i(S_t))^{1/i}$ and $k_i$ such that
 $\varphi_{i} \in (\ell_{k_i+1},
 \ell_{k_i}]$ 
   	\IF {$V_i(S_t) - V_i(S_t^+(\overline p_t; u_t)) > \ell^i_{k_i + 1} / ( 2 \cdot i! )$}
      \STATE choose $p_i$ such that $V_i(S_t) - V_i(S_t^+(p_i; u_t)) =  \ell^i_{k_i + 1} /
  ( 2 \cdot i! )$
    \ELSE
      \STATE choose $p_i = \overline p_t$
    \ENDIF
  	\STATE define $K_i = \{x \in S_t; \dot{u_t}{x} = p_i \}$ 
  	\STATE define $L_i = (V_i(K_i) / c_i)^{1/i}$ (define $L_0 = \infty$)
 \ENDFOR
 \IF{$w < 1/T$}
 	\STATE set $p_t = \underline p_t$.
 \ELSE
	\STATE  let $M(i) = \max \{j ; k_i = k_j \}$ 
	\STATE  find a $j$ such that $L_{j-1} \geq w \geq L_{M(j)}$ 
	\STATE  let $J = M(j)$ and set $p_t = p_{J}$.
 \ENDIF
\end{algorithmic}

\end{algorithm}

\begin{lemma} PricingSearch (Algorithm \ref{algo:pricingdd}) is well defined, i.e., it is always
  possible to choose $p_i$ and $j$ with the desired properties.
\end{lemma}

\begin{proof}
  For the choice of $p_i$, if $V_i(S_t) - V_i(S_t^+(\overline p_t; u_t)) >
  \ell^i_{k_i + 1} /  ( 2 \cdot i! )$, then the function 
  $\phi_i : [\underline p_t,\overline p_t] \rightarrow \R$,
  $\phi_i(p) = V_i(S_t) - V_i(S_t^+(p; u_t)) $ is continuous and
  monotone with $\phi_i(\underline p_t) = 0$ and $\phi_i(\overline p_t) >
  \frac{1}{2} i! \cdot \ell^i_{k_i + 1}$ so this
  guarantees the existence of such $p_i$.

  For the choice of $j$, let $0 = i_0 < i_1 < \hdots < i_a = d$ be the indices $i$
  such that $M(i) = i$. Notice that the intervals $[L_{i_{s+1}}, L_{i_{s}})$ are of
  the form $[L_{M(i)}, L_{i-1})$ for $i = i_s + 1$. Finally notice that the intervals $[L_{i_{s+1}}, L_{i_{s}})$ cover the entire interval $[L_d, L_0) =
  [0, \infty)$ so one of them must contain $w$.
\end{proof}

Before we proceed to the main analysis, we begin by proving a couple of lemmas regarding the ranges $(\ell_{k+1},
\ell_k]$ of the intrinsic volumes before and after each iteration. The first
lemma says that if we overprice (i.e. $S_t^-$ is chosen) the quantity $\varphi_J$ jumps from the range
$[\ell_{k_J+1},\ell_{k_J})$ to the next range $(\ell_{k_J+2},\ell_{k_J+1}]$. 

\begin{lemma}\label{lem:pricing1}
  $[ J! \cdot V_J(S^-_t(p_J;u_t)) ]^{1/J} \leq \ell_{k_J+1}$
\end{lemma}

\begin{proof}
  We abbreviate $S^-_t(p_J;u_t)$ and $S^+_t(p_J;u_t)$ by $S^-$ and $S^+$
  respectively. Using the fact that $V_J$ is a valuation and that $S^- \cap S^+
  = K_J$ we have that:
  $$V_J(S^-) = V_J(S_t) - V_J(S^+) + V_J(K_J) \leq \ell^J_{k_J + 1} / ( 2 \cdot
  J! ) + V_J(K_J)$$
  It remains to show that  $V_J(K_J) \leq \ell^J_{k_J + 1} / ( 2 \cdot J! ) $.
  To do this, we will again use the Cone Lemma to obtain the following
  inequalities:
  $$\frac{1}{J+1}V_{J}(K_J)w \leq V_{J+1}(S_t) \leq 
  \frac{1}{(J+1)!} \ell^{J+1}_{k_{(J+1)}} \leq 
  \frac{1}{(J+1)!}  \ell^{J+1}_{(k_{J})+1} $$
  The first inequality is the Cone Lemma (Lemma \ref{lem:cone}) applied to the fact that $S_t$ contains a cone of base $K_J$ and height at least $w$. The second
  inequality comes from the definition of $k_J$ and the third
  inequality comes from the fact that $J = M(J)$ so $k_{J+1} \geq k_J + 1$.

  Finally, observe that because of our choice of $J$, $w \geq L_J =
  (V_J(K_J)/c_J)^{1/J}$. Substituting in the previous equation we
  obtain:
  $$  \frac{1}{J+1}V_{J}(K_J)^{(J+1)/J}(c_J)^{-1/J} \leq
  \frac{1}{(J+1)!}  \ell_{(k_J)+1}^{J+1}$$
  Substituting the definition of $c_J$ and simplifying, we get the desired bound
  of  $V_J(K_J) \leq \ell^J_{k_J + 1} / ( 2 \cdot J! ) $.
\end{proof}

We next show that, for our chosen $J$, if we underprice, then the $J$th intrinsic volume of our knowledge set decreases by at least $\ell_{k_{J}+1}^{J}$. This will allow us to bound the number of times we can potentially underprice before $k_{J}$ changes (in particular, it is at most $2\ell_{k_J}^{J}/\ell_{k_{J}+1}^{J}$).

\begin{lemma}\label{lemma:pricing2}
  $V_J(S_t) - V_J(S_t^+(p_J; u_t)) =  \ell^J_{k_J + 1} / ( 2 \cdot J! )$
\end{lemma}

\begin{proof}
  Note that this equality is guaranteed by the algorithm's choice of $p_J$, except when $V_J(S_t) - V_J(S_t^+(\overline p_t;
  u_t)) < \ell^J_{k_J + 1} / ( 2 \cdot J! )$ and $p_J = \overline p_t$. However, in
  this case, $S_t^-(\overline p_t; u_t) = S_t$ by the definition of $\overline
  p_t$. Lemma \ref{lem:pricing1} then implies that $[ J! \cdot V_J(S_t) ]^{1/J}
  \leq \ell_{k_J+1}$, but this contradicts the definition of $k_J$.
\end{proof}

We now show that in each round, the width of the knowledge set (and thus our loss) is at most $2\ell_{k_J}$. 

\begin{lemma}\label{lemma:pricing3}
  $w \leq 2 \ell_{k_J}$
\end{lemma}

\begin{proof}
  We will derive both an upper and lower bound on $V_j(S_t)$. For the upper bound
  we again apply the Cone Lemma (Lemma \ref{lem:cone}).
  $$V_j(S_t) \geq \frac{1}{j} V_{j-1}(K_{j-1}) w \geq \frac{1}{j} (c_{j-1}
  w^{j-1}) w$$
  If $j > 1$, then the first inequality holds since $S_t$ contains a cone of
  base $K_{j-1}$ and height $w$, and the second inequality follows from the fact that $w \leq
  L_{j-1}$. If $j=1$, then we observe that $S_t$ contains a segment of length $w$,
  so $V_1(S_t) \leq w$. 
  
  To get a lower bound on $V_j(S_t)$, simply note that
  $$V_j(S_t) \leq \ell_{k_j}^j / (j!) = \ell_{k_J}^j / (j!) $$
  where the first inequality follows from the definition of $k_j$ and the second
  from the fact that $k_j = k_J$ since $J = M(j)$.

  Together the bounds imply that $c_{j-1} w^j / j \leq \ell_{k_J}^j /
  (j!)$. Substituting in the value of $c_{j-1}$ and simplifying we obtain that
  $w \leq 2 \ell_{k_J}$.
\end{proof}  

Finally, we argue that if $w$ is large enough (at least $1/T$), then $k_J$ is at most $O_d(\log\log T)$. Once $w$ is at most $1/T$, we can always price at $\underline p_t$ and incur at most $O(1)$ additional regret, so this provides a bound for the number of times we can e.g. overprice.

\begin{lemma}\label{lemma:pricing4}
  In iterations where $w \geq 1/T$, then $k_J \leq O(d \log \log (dT))$.
\end{lemma}

\begin{proof}
  It follows directly from Lemma \ref{lemma:pricing3}:
  $1/T \leq w \leq 2 \ell_{k_J} = 2 d^2 \exp(-\alpha^{k_J})$.
  Simplifying the expression we get $k_J \leq O(d \log \log (dT))$
\end{proof}

We are now ready to prove our main result:

\begin{theorem}\label{thm:pricing}
  The total loss of PricingSearch (Algorithm \ref{algo:pricingdd}) is bounded by
  $O(d^4 \log \log(dT))$.
\end{theorem}

\begin{proof}
  We sum the loss in different cases. The first is when $w < 1/T$ and the
  algorithm prices at $\underline p_t$. In those occasions the algorithm always
  sells and the loss is at most $2w \leq 2/T$, so the total loss is at most $2$.

  The second case is when the algorithm overprices and  doesn't sell. If the algorithm doesn't
  sell, then by Lemma \ref{lem:pricing1}, then $\phi_J$ goes from range
$(\ell_{k_J+1}, \ell_{k_J}]$ to the next range $(\ell_{k_J+2}, \ell_{k_J+1}]$.
Since $k_J \leq O(d \log \log (dT))$ by Lemma \ref{lemma:pricing4} this can
happen at most this many times for each index $J$. Since there are $d$ such
indices and the loss of each event is at most $1$, the total loss is
bounded by $O(d^2 \log \log (dT))$.

  The final case is when the algorithm underprices. The loss in this case is
  bounded by the width $2w$. We sum the total loss of events in which the
  algorithm overprices. We fix the selected index $J$ and $k_J$. The loss in such
  a case is at most $2w \leq 4 \ell_{k_J}$ by Lemma \ref{lemma:pricing3}.
  Whenever this happens $S_{t+1} = S_t^+(p_J; u_t)$ so the $J$-th intrinsic
  volume decreases by $\ell^J_{k_J + 1} / (2 J!)$ since $V_J(S_t) - V_J(S_{t+1}) =
  V_J(S_t) - V_J(S_t^+) = \ell^J_{k_J + 1} / (2 J!)$ by Lemma \ref{lemma:pricing2}.
  Since $V_J(S_t) \leq \ell^J_{k_J} / (J!)$. Therefore the total number of times
  it can happen is: $2 \ell^J_{k_J} / \ell^J_{k_J+1}$. The total loss is at most
  the number of times the event can happen multiplied by the maximum loss for an event,
  which is:
  $$\frac{2 \ell^J_{k_J} }{\ell^J_{k_J+1}} \cdot (4 \ell_{k_J}) = 8
  \frac{\ell^{J+1}_{k_J}}{\ell^J_{k_J+1} } = 8 d^2 \exp(J \alpha^{k_J +1} -
  (J+1) \alpha^{k_J}) \leq 8 d^2 \exp(\alpha^{k_J} (d \alpha - (d+1)) = 8 d^2  $$
  since $\alpha = 1+1/d$. By summing over all $d$ possible values of $J$ and all
  $O(d \log \log (dT))$ values of $k_J$ we obtain a total loss of
  $O(d^4 \log \log(dT))$.
\end{proof}

\subsection{Proof of the Cone Lemma}\label{sec:cone_lemma}

We will prove the Cone Lemma in three steps. We start by proving some geometric
lemmas about how linear transformations affect intrinsic volumes. We then use these lemmas to bound the intrinsic volumes of cylinders. Finally, by approximating a cone as a stack of thin cylinders, we apply these bounds to prove the Cone Lemma.

\subsubsection{Geometric lemmas}

Define an \textit{$\alpha$-stretch} of $\R^{d}$ as a linear
transformation which contracts $\R^{d}$ along some axis by a factor of $\alpha$,
leaving the remaining axes untouched (in other words, there is some coordinate
system in which an $\alpha$-stretch $T_{\alpha}$ sends $(x_1, x_2, \dots,
x_{d})$ to $(\alpha x_1, x_2, \dots, x_{d})$).

A contraction is a linear transformation $T : \R^d \rightarrow \R^d$ such
that $\norm{T x} \leq \norm{x}$ for all $x \in \R^d$. An $\alpha$-stretch is a
contraction whenever $\alpha \in [0,1]$.


\begin{lemma}\label{lem:twoplaneproj} Let $H$ and $H'$ be two ($d$-dimensional)
hyperplanes in $\mathbb{R}^{d+1}$, whose normals are separated by angle $\theta$.
Let $K$ be a convex body contained in $H$, and let $K'$ be the projection of $K$
onto $H'$. Then $K'$ is (congruent to) a $(\cos\theta)$-stretch of $K$.
\end{lemma}

\begin{proof} Without loss of generality, let $H'$ be the hyperplane
with orthonormal basis $e_1, e_2, \dots, e_d$, and let $H$ be the hyperplane
with orthonormal basis $e'_1 = (\cos\theta) e_1 + (\sin\theta) e_{d+1}, e'_2 =
e_2, \dots, e'_d = e_d$. Note that a point $a_1e'_1 + a_2e'_2 + \dots a_ne'_n$
in $H$, projects to the point $(\cos\theta)a_1e_1 + a_2e_2 + \dots + a_ne_n$ in
$H'$. This is the definition of a $(\cos \theta)$-stretch.  \end{proof}

\medskip

The next lemma bounds the change in the $d$-th volume of a $(d+1)$-dimensional
object when it is transformed by a contraction. The analysis will be based on
the fact that for a $(d+1)$-dimensional convex set $S$, $V_d(S)$ corresponds to
half of the surface area. This fact can be derived either
from Hadwiger's theorem (Theorem \ref{thm:hadwiger}) or from Cauchy's formula
for the surface area together with Theorem \ref{thm:random_projections}.

It is simpler to reason about the surface area of polyhedral convex sets (i.e.
sets that can be described as a finite intersection of half-spaces). The
boundary of a polyhedral convex set in $\R^{d+1}$ can be described as a finite
collection of facets, which are convex sets of dimension $d$. The surface area
corresponds to the sum of the $d$-dimensional volume of the facets. For a
$2$-dimensional polytope the surface area correspond to the perimeter. For a
$3$-dimensional polytope the surface area corresponds to the sum of the area
(the $2$-dimensional volume) of the facets. For a general convex set $K$, the
surface area can be computed as the limit of the surface area of $K_t$ where
$K_t$ are polyhedral sets that converge (in the Hausdorff sense) to $K$. This is
equivalent to the usual definition of the surface area as the surface integral
of a volume element.

Given the discussion in the previous paragraph, to reason about how the surface
area transforms after a linear transformation, it is enough to reason how
the volume of $d$-dimensional convex sets (the facets) transform when the
ambient $\R^{d+1}$ space is transformed by a linear transformation.

\begin{lemma}\label{lem:stretch} Let $K \in \Conv_{d+1}$ and $T$ be a
  contraction, then  $$V_d(T(K)) \geq \det T \cdot V_{d}(K)$$  \end{lemma}

\begin{proof} By the previous discussion, $V_{d}(K)$ is proportional
to the surface area of $K$. By taking finer and finer
approximations of $K$ by polytopes, it suffices to prove the result for a
polyhedral set. We only need to argue how the $d$-dimensional volume of the
facets is transfomed by $T$. The change in volume of a facet corresponds to the
determinant of the transformation induced by $T$ on the tangent space of that
  facet\footnote{The tangent space of a facet is the space of all vectors that
  are parallel to that facet}.
More precisely, given vectors linearly independent vectors $v_1, \hdots, v_d \in
\R^{d+1}$, let $P$ be the parallelepiped generated by them and let $V_d(P)$ be
its volume. Let also $n$ be the unit vector orthogonal to affine subspace
  containing $P$ and $N$
an unit segment in that direction, i.e., the set of points of the form $tn$
for $t\in [0,1]$, then:
$$V_{d+1}(T (P+N)) = (\det T) \cdot V_d(P+N) = (\det T) \cdot V_{d-1}(P)$$
  where the first equality follows from how the (standard) volume transforms and
  the second since $N$ is orthogonal to $P$ and has size $1$.
  
  Now, since $T(P+N) = T(P) + T(N)$, the volume $V_{d+1}(T(P+N))$ can be
  written as $V_{d-1}(T(P))$ times the projection of $N$ in the orthogonal
  direction of $T(P)$, which is $\dot{T n}{ n'} \leq \norm{T n} \cdot \norm{n'}
  \leq 1$ where $n'$ is the orthogonal vector to $T(P)$ and $\norm{T n} \leq 1$
  follows from the fact that $T$ is a contraction. Therefore:
  $$V_{d-1}(T(P)) \geq V_{d+1}(T (P+N)) = (\det T) \cdot V_{d-1}(P)$$

\end{proof}

\subsubsection{Intrinsic volumes of cylinders}
Given a convex set $K$ in $\R^{d}$, an \textit{orthogonal cylinder} with base
$K$ and height $w$ is the convex set in $\R^{d+1}$ formed by taking the
Minkowski sum of $K$ (embedded into $\R^{d+1}$) and a line segment of length $w$
orthogonal to $K$.

\begin{lemma}\label{lem:cylinder}
Let $K$ be a convex set in $\R^d$, and let $S$ be an orthogonal cylinder with base $K$ and height $h$. Then, for all $0 \leq j \leq d$,

$$V_{j+1}(S) = V_{j+1}(K) + hV_{j}(K).$$
\end{lemma}
\begin{proof}
Embed $K$ into $\R^{d+1}$ so that it lies in the hyperplane $x_{d+1} = 0$, and
  let $L$ be the line segment from $0$ to $he_{d+1}$, so that $S = K+L$ is an
  orthogonal cylinder with base $K$ and height $h$. We will compute
  $\Vol_{d+1}(S+\eps B_{d+1})$. Recall that $\Vol$ refers to the standard
  volume. Whenever we add subscripts (e.g. $\Vol_d$) we do so to highlight that
  we are talking about the standard volume of a convex set in a $d$-dimensional
  (sub)space.
  
  We claim we can decompose $S+\eps B_{d+1}$ into
  two parts; one with total volume $\Vol_{d+1}(K+\eps B_{d+1})$, and one with
  total volume $h\Vol_{d}(K + \eps B_{d})$. To begin, consider the intersection
  of $S + \eps B_{d+1}$ with $\{ x_{d+1} \in [0, h]\}$. We claim this set has
  volume at least $h\Vol_{d}(K + \eps B_{d})$. In particular, note that (since
  $S$ is an orthogonal cylinder) every cross-section of the form $(S + \eps
  B_{d+1}) \cap \{ x_{d+1} = t\}$ for $t \in [0, h]$ is congruent to the set $K
  + \eps B_{d}$. It follows that the volume of this region is $h\Vol_{d}(K+\eps
  B_{d})$.

Next, consider the intersection of $S + \eps B_{d+1}$ with the set $\{ x_{d+1}
  \not\in [0, h] \}$. This intersection has two components: a component $S^{+}$,
  the intersection of $S + \eps B_{d+1}$ with the set $\{ x_{d+1} \geq h \}$,
  and a component $S^{-}$, the intersection of $S + \eps B_{d+1}$ with the set
  $\{ x_{d+1} \leq 0 \}$ (see Figure \ref{fig:cylinder}). Now, define $K^{+}$ to be the intersection of $K +
  \eps B_{d+1}$ with $\{x_{d+1} \geq 0\}$, and let $K^{-}$ be the intersection
  of $K + \eps B_{d+1}$ with $\{x_{d+1} \leq 0\}$. It is straightforward to
  verify that $K^{+}$ is congruent to $S^{+}$ and that $K^{-}$ is congruent to
  $S^{-}$, and therefore the volume of this region is equal to $\Vol(K^{+}) +
  \Vol(K^{-}) = \Vol_{d+1}(K+\eps B_{d+1})$. 

We therefore have that $\Vol_{d+1}(S+\eps B_{d+1}) = \Vol_{d+1}(K+\eps B_{d+1}) + h\Vol_{d}(K + \eps B_{d})$. Expanding out all parts via Steiner's formula (\ref{eqn:steineralt}), we have that:

$$\sum_{j=0}^{d+1} \kappa_{d+1-j}V_j(S)\eps^{d+1-j} = \sum_{j=0}^{d}\kappa_{d+1-j}V_{j}(K)\eps^{d+1-j} + h\sum_{j=0}^{d}\kappa_{d-j}V_j(K)\eps^{d-j}.$$

Equating coefficients of $\eps^{d-j}$, we find that

$$V_{j+1}(S) = V_{j+1}(K) + hV_{j}(K).$$
\end{proof}

\begin{figure}
\centering
\begin{subfigure}[b]{0.40\textwidth}
  \centering
\begin{tikzpicture}
  \draw[dashed, fill=red!20!white] (-.4,0) -- (-.4,2) .. controls (-.4,2.22) and (-0.22,2.4) .. (0,2.4) --
  (1,2.4) .. controls  (1.22,2.4) and (1.4,2.22)  .. (1.4,2) -- (1.4,0) ..
  controls    (1.4,-.22)  and (1.22,-.4)  .. (1,-.4) -- (0,-.4) .. controls
  (-.22,-.4) and (-.4,-.22)  .. cycle;

  \draw[line width = 1] (0,0) -- (1,0) -- (1,2) -- (0,2) -- cycle;
  \draw[dashed] (-.4, 0) -- (1.4,0);
  \draw[dashed] (-.4, 2) -- (1.4,2);
  \node at (.5,2.20) {\tiny $S^+$};
  \node at (.5,-.2) {\tiny $S^-$};
  \begin{scope}[shift={(2.5,1 )}]
   \draw[dashed, fill=red!20!white] (-.4,0) .. controls (-.4,0.22) and
   (-0.22,0.4) .. (0,0.4) --
  (1,0.4) .. controls  (1.22,0.4) and (1.4,0.22)  .. (1.4,0) -- (1.4,0) ..
  controls    (1.4,-.22)  and (1.22,-.4)  .. (1,-.4) -- (0,-.4) .. controls
  (-.22,-.4) and (-.4,-.22)  .. cycle;
  \draw[line width = 1](0,0) -- (1,0);
    \draw[dashed] (-.4,0) -- (1.4, 0);
\node at (.5,.2) {\tiny $S^+$};
  \node at (.5,-.2) {\tiny $S^-$};
  \end{scope}
\end{tikzpicture}
  \caption{Proof of Lemma \ref{lem:cylinder}: we decompose $K+L+\epsilon B$
  (left) and $K + \epsilon B$ (right).
  }
\label{fig:cylinder}
\end{subfigure}
\qquad
\begin{subfigure}[b]{0.40\textwidth}
  \centering
\begin{tikzpicture}
  \draw[line width = 1, fill=red!20!white] (0,0) -- (1,0) -- (1,2) -- (0,2) -- cycle;
  \node at (.5,1) {$S_\bot$};
  \begin{scope}[shift={(2,0 )}]
  \draw[line width = 1, fill=red!20!white] (0,0) -- (1,0) -- (2,2) -- (1,2) -- cycle;
  \draw[line width = 1.5, blue] (3/5, 6/5) -- (7/5, 4/5);
    \node at (.75,.5) {$S_1$};
    \node at (1.25,1.5) {$S_2$};
  \end{scope}
  \begin{scope}[shift={(5,1 )}]
  \draw[line width = 1, fill=red!20!white] (0,0) -- (1,0) -- (7/5, 4/5) -- (3/5,
    6/5)  -- cycle;
    \draw[line width = 1.5, blue] (3/5, 6/5) -- (7/5, 4/5);
    \node at (.75,.5) {$S_1$};
  \end{scope}
  \begin{scope}[shift={(4,-1 )}]
  \draw[line width = 1, fill=red!20!white] (3/5, 6/5) -- (7/5, 4/5) -- (2,2) -- (1,2) -- cycle;
    \draw[line width = 1.5, blue] (3/5, 6/5) -- (7/5, 4/5);
    \node at (1.25,1.5) {$S_2$};
  \end{scope}
\end{tikzpicture}
  \caption{First step in the proof of Lemma \ref{lem:oblique}, we cut and
  re-assemble a cylinder  }
\label{fig:oblique}
\end{subfigure}

\begin{subfigure}[b]{0.40\textwidth}
  \centering
\begin{tikzpicture}
  \draw[line width = 1, fill=red!20!white] (0,0) -- (1,0) -- (1+1/10,2/10) --
  (1/10,2/10) -- cycle;
  \draw[line width = 1.5, blue] (1/5, 6/5-4.5/5) -- (5/5, 4/5-4.5/5);
  \node at (1.3,.1) {\small $S$};

  \begin{scope}[shift={(2.5,0 )}]
  \draw[line width = 1, fill=red!20!white] (0,0) -- (1,0) -- (1+4/10,2*4/10) --
  (4/10,2*4/10) -- cycle;
  \draw[line width = 1, fill=red!20!white] (0,0) -- (1,0) -- (1+3/10,2*3/10) --
  (3/10,2*3/10) -- cycle;
  \draw[line width = 1, fill=red!20!white] (0,0) -- (1,0) -- (1+2/10,2*2/10) --
  (2/10,2*2/10) -- cycle;
  \draw[line width = 1, fill=red!20!white] (0,0) -- (1,0) -- (1+1/10,2*1/10) --
  (1/10,2*1/10) -- cycle;
    \node at (1.4,.1) {\small $S^{[4]}$};
  \draw[line width = 1.5, blue] (3/5-3/10, 6/5-3/5) -- (7/5-3/10, 4/5-3/5);
    \end{scope}
  \begin{scope}[shift={(5,0 )}]
  \draw[line width = 1, fill=red!20!white] (0,0) -- (1,0) -- (1,2*4/10) --
  (0,2*4/10) -- cycle;
  \draw[line width = 1, fill=red!20!white] (0,0) -- (1,0) -- (1,2*3/10) --
  (0,2*3/10) -- cycle;
  \draw[line width = 1, fill=red!20!white] (0,0) -- (1,0) -- (1,2*2/10) --
  (0,2*2/10) -- cycle;
  \draw[line width = 1, fill=red!20!white] (0,0) -- (1,0) -- (1,2*1/10) --
  (0,2*1/10) -- cycle;
    \node at (1.4,.1) {\small $S^{[4]}_\perp$};
    \end{scope}

\end{tikzpicture}
  \caption{Stacking thin oblique and orthogonal cylinders in Lemma
  \ref{lem:oblique}}
\label{fig:stacking}
\end{subfigure}
\quad
\begin{subfigure}[b]{0.40\textwidth}
  \centering
\begin{tikzpicture}
  \draw[line width = 1, fill=red!20!white] (0,0) -- (3,2) -- (2,0) -- cycle;
  \draw[line width = 1] (.5*1.5,.5) -- (2+.5*.5,.5);
  \draw[line width = 1] (1*1.5,1) -- (2+.5,1);
  \draw[line width = 1] (1.5*1.5,1.5) -- (2+.5*1.5,1.5);
  \node at (1.5+.2,.5+.25) {\small $S_2$};
  \node at (2+.2,1+.25) {\small $S_1$};
  \node at (2.5+.2,1.5+.25) {\small $S_0$};
  \node at (1+.2,.25) {\small $S_3$};

  \begin{scope}[shift={(3.5,0 )}]
  \draw[dashed] (0,0) -- (3,2) -- (2,0) -- cycle;
  \draw[line width = 1, fill=red!20!white] (.5*1.5,.5) -- (2+.5*.5,.5) --
    (2+.5*.5-.25,0) --  (.5*1.5-.25,0) -- cycle;
  \draw[line width = 1, fill=red!20!white] (1*1.5,1) -- (2+.5,1) -- (2+.25,.5) -- (1*1.5-.25,.5) --
    cycle ;
  \draw[line width = 1, fill=red!20!white] (1.5*1.5,1.5) -- (2+.5*1.5,1.5) --  (2+.5*1.5-.25,1)
    -- (1.5*1.5-.25,1) -- cycle;
  \end{scope}
\end{tikzpicture}
  \caption{Approximating a cone by thin oblique cylinders}
\label{fig:cone_cylinder}
\end{subfigure}

\caption{Illustration of the cylinder and cone proofs. In all cases, the $x$-axis is a
$d$-dimensional space and the $y$-axis a $1$-dimensional space }
\end{figure}
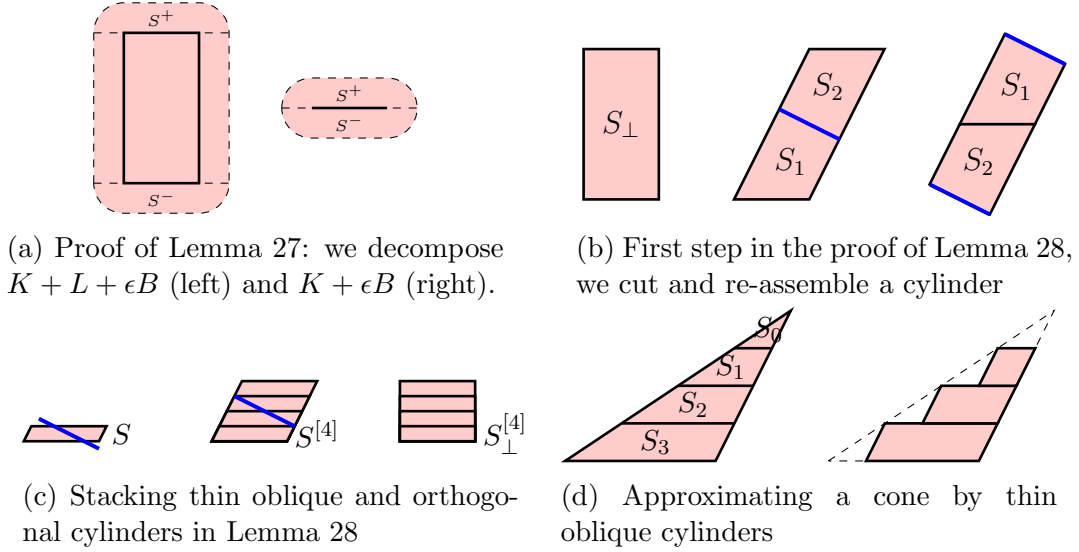

An \textit{oblique cylinder} in $\R^{d+1}$ is formed by taking the Minkowski sum
of a convex set $K \subset \R^{d}$ and a line segment $L$ not necessarily
perpendicular to $K$. The \textit{height} of an oblique cylinder is equal to the
length of the component of $L$ orthogonal to the affine subspace containing $K$.

\begin{lemma}\label{lem:oblique}
Let $K$ be a convex set in $\R^d$. If $S$ is an oblique cylinder with base $K$ and height $h$, and $S_{\perp}$ is an orthogonal cylinder with base $K$ and height $h$, then (for all $1 \leq j \leq d+1$)

$$V_{j}(S) \geq V_{j}(S_{\perp}).$$
\end{lemma}

\begin{proof}
Note that when $j=d+1$, $V_{j}(S) = V_{j}(S_{\perp})$ as they are related by a
linear transformation with determinant $1$. For the remaining cases, we will
  first prove for $j=d$ and then reduce all other cases to $j=d$.
 
  \paragraph{Case $\mathbf{j=d}$ (tall cylinder).}
Write $S = K + L$, where $L$ is a line segment of length $\ell$ (with orthogonal
  component $h$ with respect to $K$). We will begin by choosing a
  hyperplane $H$ perpendicular to $L$ that intersects $S$ along its lateral
  surface, dividing it into two sections $S_1$ and $S_2$ (see Figure
  \ref{fig:oblique}).
  Note that this is only possible if the height $h$ of
  this cylinder is large enough with respect to the diameter of $K$ and angle
  $L$ makes with $K$. We address the case of the short cylinder in the next case.
  Let $K' = H \cap S$. By Theorem \ref{thm:valuation}, we know that $V_{d}(S) =
  V_{d}(S_1) + V_{d}(S_2) - V_d(K')$. 

Note that it is possible to reassemble $S_1$ and $S_2$ by gluing them along
  their copies of $K$ to form an orthogonal cylinder with base $K'$ and height
  $\ell$. Call this cylinder $S'$. Again by Theorem \ref{thm:valuation}, we have
  that $V_{d}(S') = V_{d}(S_1) + V_{d}(S_2) - V_d(K)$, and therefore $V_{d}(S) =
  V_{d}(S') + V_{d}(K) - V_{d}(K')$. But by Lemma \ref{lem:cylinder}, $V_{d}(S')
  = V_{d}(K') + \ell V_{d-1}(K')$, so $V_{d}(S) = V_{d}(K) + \ell V_{d-1}(K')$.
  On the other hand (also by Lemma \ref{lem:cylinder}), $V_{d}(S_{\perp}) =
  V_{d}(K) + h V_{d-1}(K)$. Therefore, to show that $V_{d}(S) \geq
  V_{d}(S_{\perp})$, it suffices to show that $\ell V_{d-1}(K') \geq h
  V_{d-1}(K)$. 

  Now, note that $K'$ is the projection of $K$ onto the hyperplane $H$. The normal
  to $H$ is parallel to $L$. Since $L$ has length $\ell$ and orthogonal
  component $h$ with respect to $K$, the angle between $L$ and the normal
  to $K$ equals $\arccos(h/\ell)$, from which it follows from Lemma
  \ref{lem:twoplaneproj} that $K'$ is an $(h/\ell)$ stretch of $K$. By Lemma
  \ref{lem:stretch}, it follows that $V_{j-1}(K') \geq (h/\ell)V_{j-1}(K)$, from
  which the desired inequality follows.

  \paragraph{Case $\mathbf{j=d}$ (short cylinder).}
Finally, what if the original cylinder was not tall enough to divide into two
  components in the desired manner? To deal with this, let $S^{[n]}$ denote $n$
  copies of $S$ stacked on top of each other (i.e. $S^{[n]} = K + nL$), and let
  $S_{\perp}^{[n]}$ denote $n$ copies of $S_{\perp}$ stacked on top of each
  other (i.e. an orthogonal cylinder with base $K$ and height $nh$). Repeatedly
  applying Theorem \ref{thm:valuation}, we have that $V_{d}(S^{[n]}) = nV_{d}(S)
  + (n-1)V_{d}(K)$, and that $V_{d}(S_{\perp}^{[n]}) = nV_{d}(S_{\perp}) +
  (n-1)V_{d}(K)$. Therefore, to show that $V_{d}(S) \geq V_{d}(S_{\perp})$, it
  suffices to show that $V_{d}(S^{[n]}) \geq V_{d}(S_{\perp}^{[n]})$. For some
  $n$, $S^{[n]}$ will be tall enough to divide as desired, which completes the
  proof.

\paragraph{Reducing $\mathbf{j<d}$ to $\mathbf{j=d}$.}
  We can without loss of generality assume that $K$ (the base of the cylinder)
  is in the plane spanned by the first $d$ coordinate vectors. Also, let $\pi_d :
  \R^{d+1} \rightarrow \R^d$ be the projection in the first $d$ coordinates. 

  Recall that the $j$th intrinsic volume $V_{j}(S)$ is equal to the expected
  volume of the projection of $S$ onto a randomly chosen $j$-dimensional
  subspace of $\R^{d+1}$, where the distribution over subspaces is given by the
  Haar measure over  $\Gr(d+1, k)$ (see Theorem \ref{thm:random_projections}). 

  Therefore, choose $H$ according to this measure and let $P = \pi_d(H)$.
  Note that (almost surely) $P$ is an element of $\Gr(d, j)$
  and $P$ is distributed according to the Haar
  measure of this Grassmannian. By the law of total expectation, we can write
\begin{eqnarray*}
V_{j}(S) &=& \E_{H \sim \Gr(d+1, j)}\left[V_{j}(\Pi_{H}S)\right] 
= \E_{P \sim \Gr(d, j)}\left[\E_{H\sim \Gr(d+1, j)}\left[V_{j}(\Pi_{H}S) \,|\, \pi_{d}(H) = P\right]\right]
\end{eqnarray*}

Let $P'$ be the element of $\Gr(d+1, j+1)$ spanned by $P$ and $e_{d+1}$. Note that since $H \subset P'$, $\Pi_{H}S = \Pi_{H}\Pi_{P'}S$. We therefore claim that

$$\E_{H\sim \Gr(d+1, j)}\left[V_{j}(\Pi_{H}S) \,|\, \pi_{d}(H) = P\right] = V_{j}(\Pi_{P'}S).$$

Indeed, conditioned on $\pi_{d}(H) = P$, $H$ is a (Haar-)uniform subspace of
dimension $j$ of the $j+1$-dimensional space $P'$, from which
the above equality follows. Therefore, we have that
$$V_{j}(S) = \E_{P \sim \Gr(d, j)}\left[V_{j}(\Pi_{P'}S)\right]$$

\noindent
and similarly

$$V_{j}(S_{\perp}) = \E_{P \sim \Gr(d, j)}\left[V_{j}(\Pi_{P'}S_{\perp})\right].$$

Now, since $e_{d+1}$ belongs to $P'$, if $S_{\perp}$ is an orthogonal cylinder
with base $K$ and height $h$ in $\R^{d+1}$, then $\Pi_{P'}S_{\perp}$ is an
orthogonal cylinder with base $\Pi_{P'}K$ and height $h$ in $P'$. Likewise,
$\Pi_{P'}S$ is an oblique cylinder with base $\Pi_{P'}K$ and height $h$ in $P'$.
Since $P'$ is $j+1$ dimensional, it follows the previous cases that:
$V_{j}(\Pi_{P'}S) \geq V_{j}(\Pi_{P'}S_{\perp})$
so

$$V_{j}(S) = \E_{P \sim \Gr(\d, j)}\left[V_{j}(\Pi_{P'}S)\right] \geq \E_{P \sim
\Gr(d, j)}\left[V_{j}(\Pi_{P'}S_{\perp})\right] \geq V_{j}(S_{\perp})$$
\end{proof}

\subsubsection{Intrinsic volumes of cones}

A \textit{cone} in $\R^{d+1}$ is the convex hull of a $d$-dimensional convex set
$K$ and a point $p \in \R^{d+1}$. If the distance from $p$ to the affine
subspace containing $K$ is $h$, we say the cone has \textit{height} $h$
and \textit{base} $K$.

\begin{lemma}\label{lem:cone2}
Let $K$ be a convex set in $\R^{d}$, and let $S$ be a cone in $\R^{d+1}$ with base $K$ and height $h$. Then, for all $0 \leq j \leq d$,

$$V_{j+1}(S) \geq \frac{1}{j+1}h V_{j}(K).$$
\end{lemma}
\begin{proof}

Choose a positive integer $n$, and divide $S$ into $n$ parts via the hyperplanes
  $H_{i} = \{ x_{d+1} = \frac{n-i}{n}h\}$ (for $0 \leq i \leq n$). For $0 \leq i
  < n$, let $K_i$ be the intersection of $H_i$ with $S$, and let and let $S_i$
  be the region of $S$ bounded between hyperplanes $H_{i}$ and $H_{i+1}$ (see
  Figure \ref{fig:cone_cylinder}). Note
  that each $S_i$ is a frustum with bases $K_i$ and $K_{i+1}$ and height $h/n$,
  and furthermore that each $K_i$ is congruent to $\frac{i}{n}K$. 

By repeatedly applying Theorem \ref{thm:valuation}, we know that

$$V_{j+1}(S) = \sum_{i=0}^{n-1} V_{j+1}(S_i) - \sum_{i=1}^{n-1} V_{j+1}(K_i).$$

Note that each set $S_i$ contains an oblique cylinder with base $K_{i}$ (since
  $K_{i}$ is a contraction of $K_{i+1}$, some translate of $K_{i}$ is strictly
  contained inside $K_{i+1}$) and height $h/n$. It follows from Lemmas
  \ref{lem:cylinder} and \ref{lem:oblique} that $V_{j+1}(S_i) \geq V_{j+1}(K_i)
  + \frac{h}{n}V_{j}(K_i)$. It follows that

\begin{eqnarray*}
V_{j+1}(S) \geq  \sum_{i=1}^{n-1} \frac{h}{n}V_{j}(K_i) 
= \sum_{i=1}^{n-1} \frac{h}{n}V_{j}\left(\frac{i}{n}K\right) 
= \sum_{i=1}^{n-1} \frac{h}{n}\left(\frac{i}{n}\right)^{j}V_{j}(K) 
=\left(\sum_{i=1}^{n-1}\left(\frac{i}{n}\right)^{j}\frac{1}{n}\right)h V_{j}(K).
\end{eqnarray*}

As $n$ goes to infinity, this sum approaches $\int_{0}^{1}x^{j}dx = \frac{1}{j+1}$, and therefore we have that $V_{j+1}(S) \geq \frac{1}{j+1}hV_{j}(K)$. 

\end{proof}

\subsection{Efficient implementation}\label{sec:efficient}

We have thus far ignored issues of computational efficiency. In this subsection, we will show that algorithms SymmetricSearch (Algorithm \ref{algo:symmetricdd}) and PricingSearch (Algorithm \ref{algo:pricingdd}) can be implemented in polynomial time by a randomized algorithm that succeeds with high probability.

The main primitive we require to implement both algorithms is a way to efficiently compute the intrinsic volumes of a convex set (and in particular a convex polytope, since our knowledge set starts as $[0,1]^d$ and always remains a convex polytope). Unfortunately, even computing the ordinary volume of a convex polytope (presented as an intersection of half-spaces) is known to be $\# P$-hard \cite{barany1987computing}. Fortunately, there exist efficient randomized algorithms to compute arbitrarily good multiplicative approximations of the volume of a convex set.

\begin{theorem}[Dyer, Frieze, and Kannan \cite{dyer1991random}]\label{thm:volumeapprox}
Let $K$ be a convex subset of $\R^{d}$ with an efficient membership oracle (which given a point, returns whether or not $x \in K$). Then there exists a randomized algorithm which, given input $\eps > 0$, runs in time $\poly(d, \frac{1}{\eps})$ and outputs an $\eps$-approximation to $\Vol(K)$ with high probability.
\end{theorem}

We will show how we can extend this to efficiently compute (approximately, with high probability) the intrinsic volumes of a convex polytope presented as an intersection of half-spaces. 

\begin{theorem}\label{thm:intrinsicvolumeapprox}
Let $K$ be a polytope in $\R^{d}$ defined by the intersection of $n$ half-spaces and contained in $[0,1]^d$. Then there exists a randomized algorithm which, given input $\eps > 0$ and $1 \leq i \leq d$, runs in time $\poly(d, n, \frac{1}{\eps})$, and outputs an $\eps$-approximation to $V_{i}(K)$ with high probability.
\end{theorem}
\begin{proof}
We use the fact (Theorem \ref{thm:random_projections}) that $V_i(K)$ is the expected volume of the projection of $K$ onto a randomly chosen $i$-dimensional subspace (sampled according to the Haar measure). Since $K$ is contained inside $[0,1]^d$, any $i$-dimensional projection of $K$ will be contained within an $i$-dimensional projection of $[0,1]^d$, whose $i$-dimensional volume is at most $\poly(d)$. By Hoeffding's inequality, we can therefore obtain an $\eps$-approximation to $V_{i}(K)$ by taking the average of $\poly(d, \frac{1}{\eps})$ $(\eps/2)$-approximations for volumes of projections of $K$ onto $i$-dimensional subspaces.

To approximately compute the volume of a projection of $K$ onto an $i$-dimensional subspace $S$, we will apply Theorem \ref{thm:random_projections}. Note that we can check whether a point belongs in the projection of $K$ into $S$ by solving an LP (the point adds $i$ additional linear constraints to the constraints defining $K$). This can be done efficiently in polynomial time, and therefore we have a polynomial-time membership oracle for this subproblem.
\end{proof}

We now briefly argue that Theorem \ref{thm:intrinsicvolumeapprox} allows us to implement efficient randomized variants of SymmetricSearch and PricingSearch which succeed with high probability. To do this, it suffices to note that all of the analysis of both algorithms is robust to tiny perturbations in computations of intrinsic volumes. For example, in SymmetricSearch the analysis carries through even if instead of $K_i$ dividing $S_t$ into two regions such that $V_{i}(S^{+}) = V_{i}(S^{-})$, it divides them into regions satisfying $V_{i}(S^{+}) \in [(1-\eps)V_{i}(S^{-}), (1+\eps)V_{i}(S^{-})]$ for some constant $\eps$. 

The only remaining implementation detail is how to hyperplanes $K_i$ that divide the $i$th intrinsic volume of $S_i$ equally (or in the case of PricingSearch, divide off a fixed amount of intrinsic volume). Since intrinsic volumes are monotone (Theorem \ref{thm:valuation}), this can be accomplished via binary search.

\section{Halving algorithms}\label{sect:halving}

There are many simple algorithms one can try for the contextual search problem,
like algorithms that always halve the width or volume of the current knowledge
set. One natural question is whether these simple algorithms suffice to give the
same sort of regret bounds as our algorithms based on intrinsic volumes (e.g.
SymmetricSearch). 

In this section, we show that while these algorithms also obtain $O_d(1)$
regret for the contextual search problem with symmetric loss,
the dependence on $d$ is exponential rather than polynomial (and moreover this
dependence is tight, at least for the algorithm which always divides the width
in half). Moreover, we show that even for these simpler algorithms, looking at
how intrinsic volumes of the knowledge set change is a valuable technique for
bounding the total regret.


\subsection{Dividing the width in half}

In this section, we will analyze the algorithm which always cuts the width in
half; that is, always guesses $p_t = p_{t}^{\mid} = \frac{1}{2}(\overline{p}_t +
\underline{p}_t)$. We will show that for the symmetric loss function, this
strategy achieves $2^{O(d)} = O_d(1)$ regret.

Our analysis will proceed similarly to the proof of Theorem \ref{thm:symsearch}.
We will rely on the following lemma, which shows that dividing the width in half
guarantees that the ratio of the intrinsic volume of the smaller half to that of
the larger half is still lower-bounded by some function of $d$.

\begin{lemma}\label{lem:splitratio}
If $p_{t} = p_{t}^{\mid}$, then

  $$V_{j}(S^{-}) \geq 2^{-j}V_{j}(S^{+}) \quad \text{and} \quad
  V_{j}(S^{+}) \geq 2^{-j}V_{j}(S^{-})$$
\end{lemma}
\begin{proof}
  Let $w = \frac{1}{2}( \overline{p}_t - \underline{p}_t )$ and $K$ be the
  intersection of the hyperplan $\langle x, u_t \rangle = p_t$ with $S_t$.
  Choose a point $q$ in $S_t$ such that  $\langle q, u_t \rangle =
  \underline{p}_t $ as depicted in Figure \ref{fig:grun_i}.
  Consider the cone $C$ formed by the convex hull of $q$ and
  $K$ (Figure \ref{fig:grun_ii}). Since $S^{-}$ is convex, $C$ is contained in $S^{-}$, and thus
  $V_{j}(S^{-}) \geq V_j(C)$.

  Now, consider the dilation of the cone $C$ by a factor of $2$ about the point
  $q$ (Figure \ref{fig:grun_iii}). This results in a new cone $2C$. We claim that this cone contains
  $S^{+}$. To see this, it suffices to note that the contraction of $S^{+}$ by a
  factor of $1/2$ about $q$ lies within $S^{-}$. This follows from the fact that
  $S$ is convex, and the width of $S^{-}$ is equal to the width of $S^{+}$ (so
  any segment connecting $q$ to some point $q' \in S^{+}$ has at least as much
  length in $S^{-}$ than $S^{+}$).

  It follows that $2^{j}V_{j}(C) = V_{j}(2C) \geq V_{j}(S^{+})$. Combining this
  with our earlier inequality, the first result follows. The second result follows
  symmetrically.  \end{proof}

\begin{figure}
\centering
\begin{subfigure}[b]{0.30\textwidth}
\begin{tikzpicture}[scale=.9, xscale=1.4]
  \fill[blue!0!white] (-.1,-.4) rectangle (3.4,2.6);
  \draw[line width=1.5pt] (0,1) .. controls (0,1.4) and (.6, 2) .. (1,2)
              .. controls (1.4,2) and (3,1.4) .. (3,1)
              .. controls (3,0.6) and (1.4,0) .. (1,0)
              .. controls (.6,0) and (0,.6) .. (0,1);
 \node (X1) at (1.5,2.05) {};
 \node (X2) at (1.5,-.05) {};
  \node [shape=circle, fill=black,inner sep=1.5pt,label=left:$q$] (Q1) at
 (0,1) {};
  \draw[line width=1.2pt, color=blue] (1.5,1.9)--(1.5,.1);
 \node at (1.7,1) {$K$};
 \node at (.9,1) {$S^-$};
 \node at (2.3,1) {$S^+$};
 \begin{scope}[line width=1.0pt]
 \begin{scope}[>=latex]
 \draw[<->] (0,2.5)--(1.5,2.5);
 \draw[<->] (3,2.5)--(1.5,2.5);
 \end{scope}
 \draw (0,2.4)--(0,2.6);
 \draw (1.5,2.4)--(1.5,2.6);
 \draw (3, 2.4)--(3,2.6);
 \end{scope}
 \node at (1.5/2, 2.7) {$w$};
 \node at (1.5 + 1.5/2, 2.7) {$w$};
\end{tikzpicture}
  \caption{midpoint cut}
\label{fig:grun_i}
\end{subfigure}
  \begin{subfigure}[b]{0.30\textwidth}
\begin{tikzpicture}[scale=.9, xscale=1.4]
  \fill[blue!0!white] (-.1,-.4) rectangle (3.4,2.6);
  \draw[line width=1.5pt] (0,1) .. controls (0,1.4) and (.6, 2) .. (1,2)
              .. controls (1.4,2) and (3,1.4) .. (3,1)
              .. controls (3,0.6) and (1.4,0) .. (1,0)
              .. controls (.6,0) and (0,.6) .. (0,1);
 \node (X1) at (1.5,2.05) {};
 \node (X2) at (1.5,-.05) {};
  \node [shape=circle, fill=black,inner sep=1.5pt,label=left:$q$] (Q1) at
 (0,1) {};
  \draw[dashed, fill=blue!20!white] (0,1) -- (1.5,.1) -- (1.5,1.9) -- cycle;
  \draw[line width=1.2pt, color=blue] (1.5,1.9)--(1.5,.1);
 \node at (1.7,1) {$K$};
 \node at (1,1) {$C$};
\end{tikzpicture}
\caption{$C \subseteq S^-$}
\label{fig:grun_ii}
\end{subfigure}
  \begin{subfigure}[b]{0.30\textwidth}
\begin{tikzpicture}[scale=.9, xscale=1.4]
  \fill[blue!0!white] (-.1,-.4) rectangle (3.4,2.6);
  \draw[dashed, fill=blue!20!white] (0,1) -- (3,1-2*.9) -- (3,1+2*.9) -- cycle;
  \draw[line width=1.5pt] (0,1) .. controls (0,1.4) and (.6, 2) .. (1,2)
              .. controls (1.4,2) and (3,1.4) .. (3,1)
              .. controls (3,0.6) and (1.4,0) .. (1,0)
              .. controls (.6,0) and (0,.6) .. (0,1);
 \node (X1) at (1.5,2.05) {};
 \node (X2) at (1.5,-.05) {};
  \node [shape=circle, fill=black,inner sep=1.5pt,label=left:$q$] (Q1) at
 (0,1) {};
  \draw[line width=1.2pt, color=blue] (1.5,1.9)--(1.5,.1);
 \node at (2.5,0) {$2C$};
\end{tikzpicture}
\caption{$S^+ \subseteq 2C$}
\label{fig:grun_iii}
\end{subfigure}
\caption{}
\label{fig:grunbaum} 
\end{figure}

\begin{theorem}\label{thm:width_half}
  The algorithm that always sets $p_t = p_t^\mid$ has
  regret bounded by $2^{O(d)}$ for the symmetric loss.
\end{theorem}
\begin{proof}
We will proceed similarly to the analysis of the SymmetricSearch algorithm.
  For any fixed round $t$, let $w = \frac{1}{2}\wid(S_t;u_t)$ and $K$ be
  the intersection of the hyperplane $\langle x, u_t\rangle = p_t^\mid$ with $S_t$. 

Define a sequence of constants $c_i$ so that $c_0 = 1$ and $c_i/c_{i-1} =
  2^{-(i+1)}/i$ (in other words, $c_i = 2^{-(i+1)(i+2)/2}/i!$). For $1 \leq i
  \leq d$, define $L_i = (V_{i}(K)/c_i)^{1/i}$, and let $L_0 = \infty$. Choose
  $j$ so that $L_{j-1} \geq w \geq L_j$ (it is always possible to do this by the
  same logic as in Theorem \ref{thm:symsearch}). 
Note that the proof of Lemma \ref{lemma:symm1} carries over verbatim to show
  that $V_{j}(S_t) \geq \frac{1}{j}c_{j-1}w^j$.  We now proceed in two steps:\\

  \emph{Step 1} First we show that $V_{j}(S_{t+1}) \leq (1 - 2^{-(j+2)})V_{j}(S_t)$. 

  The set $S_{t+1}$ is either $S^{+}$ or $S^{-}$. By Lemma \ref{lem:splitratio},
  we therefore have that $$V_{j}(S_{t+1}) \leq \frac{1}{1+2^{-j}}(V_{j}(S^{+}) +
  V_{j}(S^{-})).$$

  Now, since $V_j$ is a valuation, $V_{j}(S^{+}) + V_{j}(S^{-}) = V_{j}(S_t) +
  V_{j}(K)$. Now, $V_{j}(K) = c_j L_j^j \leq c_j w^j$ by the choice of $j$.
   Combining this with
  the fact $V_{j}(S_t) \geq \frac{1}{j}c_{j-1}w^j$, we observe that $V_{j}(K)
  \leq j\frac{c_j}{c_{j-1}}V_{j}(S_t)$, and therefore that

$$V_{j}(S^+) + V_j(S^{-}) \leq \left(1 + j\frac{c_{j}}{c_{j-1}}\right)V_j(S_t) =
  (1 + 2^{-(j+1)})V_{j}(S_t).$$

It follows that

$$V_{j}(S_{t+1}) \leq \frac{1}{1+2^{-j}} \cdot (1 + 2^{-(j+1)}) V_{j}(S_t) \leq
  (1 - 2^{-(j+2)})V_{j}(S_t).$$\\

  \emph{Step 2:} Next we consider the potential function $\Phi(t) =
  \sum_{i=1}^{d} V_{i}(S_t)^{1/i}$. Note that $\Phi(0) = \mathrm{poly}(d)$. We
  will show that each round, $\Phi(t)$ decreases by at least $2^{-O(d)}w$. Since
  the loss each round is upper bounded by $w$, this proves our theorem.
\begin{eqnarray*}
V_{j}(S_{t})^{1/j} - V_{j}(S_{t+1})^{1/j} &\geq & (1 - (1-2^{-(j+2)})^{1/j})V_{j}(S_t)^{1/j} 
\geq  \frac{1}{j}2^{-(j+2)}V_{j}(S_t)^{1/j} \\
&\geq & \frac{1}{j}2^{-(j+2)}\left(\frac{c_{j-1}}{j}\right)^{1/j}w 
\geq  2^{-O(d)}w.
\end{eqnarray*}
\end{proof}

The exponential dependency on the dimension is tight, as it is shown in an
example by Cohen et al \cite{CohenLL16}. 

\begin{theorem}
There is an instance of the contextual search problem with symmetric loss such that the algorithm that always sets $p_t = p_{t}^{mid}$ incurs regret $2^{\Omega(d)}$.
\end{theorem}
\begin{proof}
See \cite{CohenLL16}. The instance they describe is as follows: let the dimension $d$ be a multiple of $8$ and $v = 0 \in [0,1]^d$.
Let $X_t$ be iid random subsets of $[d]$ of size $d/4$. Now, consider feature
vectors of the form $u_t = {\bf 1}\{X_t\} / \sqrt{d/4}$ where ${\bf 1}\{X_t\}$
is the indicator vector of $X_t$. By standard concentration bounds we have that
with high probability for any $s < t < 2^{\Omega(d)}$ we will have
$\vert X_s \cap X_t \vert \leq d/8$. Therefore for such $s < t$, 
$\dot{{\bf 1}\{X_t\}}{u_s}\leq \frac{1}{2} \dot{{\bf 1}\{X_s\}}{u_s} $ and hence
${\bf 1}\{X_t\} \in S_t$ where $S_t$ is the knowledge set in step $t$. It
implies that the loss in the $t$-th step is at least $\Omega(1)$. Since there are
$2^{\Omega(d)}$ such steps, the loss grows exponentially in $d$. 
\end{proof}

%
\subsection{Dividing the volume in half}

In this section we will consider the algorithm which always divides the volume of our current knowledge set in half. More specifically, this algorithm always chooses $p_t$ so that $\Vol(S^{+}) = \Vol(S^{-})$. We will show that for the symmetric loss function, this strategy also achieves $2^{O(d)} = O_{d}(1)$ regret.

Like in the previous subsection, we will argue that splitting the volume in half
guarantees that the other intrinsic volumes are split in some ratio bounded away
from 0 and 1. To do this, we will first show that splitting the volume in half
imposes constraints on the ratio of the widths of $S^{+}$ and $S^{-}$, and then
adapt the proof of Lemma \ref{lem:splitratio}. The proof follows the same scheme
depicted in Figure \ref{fig:grunbaum} but with unequal widths for $S^+$ and $S^-$.

\begin{lemma}\label{lem:splitvol1}
Assume $p_t$ is chosen so that $\Vol(S^{+}) = \Vol(S^{-})$. Then if $w^{+} = \wid(S^{+}; u_t)$ and $w^{-} = \wid(S^{-}; u_t)$, 

  $$w^{-} \geq (2^{1/d} - 1)w^{+} \quad \text{and} \quad
w^{+} \geq (2^{1/d} - 1)w^{-}.$$
\end{lemma}
\begin{proof}
Choose a point $q$ in $S^{-}$ that is distance $w^{-}$ from the hyperplane $\langle x, u_t \rangle = p_t$. Let $K$ be the intersection of this hyperplane with the original set $S_t$. Consider the cone $C$ formed by the convex hull of $q$ and $K$. Since $S^{-}$ is convex, $C$ is contained in $S^{-}$, and therefore $\Vol(S^{-}) \geq \Vol(C)$. 

Now, consider the dilation of the cone $C$ by a factor of $\alpha = (w^{+} +
  w^{-})/w^{-}$ about the point $q$. By similar logic as in the proof of Lemma
  \ref{lem:splitratio}, this cone $\alpha C$ contains $S^{+}$. In fact, since
  $C$ is contained in $\alpha C$ (since $\alpha > 1$) and since $C$ is contained
  in $S^{-}$ (which has zero volume intersection with  $S^{+}$), $S^{+}$ is contained in $\alpha C \setminus C$. We therefore have that

$$\Vol(\alpha C) - \Vol(C) \geq \Vol(S^{+}) = \Vol(S^{-}) \geq \Vol(C).$$

Since $\Vol(\alpha C) = \alpha^d \Vol(C)$, this implies that $\alpha^{d} \geq
2$, and therefore that $(w^{+} + w^{-})/w^{-} = \alpha \geq 2^{1/d}$ and $w^{+}/w^{-} \geq 2^{1/d} - 1$,
as desired. The other inequality follows by symmetry.
\end{proof}

\begin{lemma}\label{lem:splitvol2}
Let $w^{+} = \wid(S^{+}; u_t)$ and $w^{-} = \wid(S^{-}; u_t)$. Assume $p_t$ is chosen so that $w^{+} \geq \alpha w^{-}$ and $w^{-} \geq \alpha w^{+}$, for some $\alpha > 0$. Then, for all $1 \leq j \leq d$,

$$V_{j}(S^{-}) \geq \left(1 + \frac{1}{\alpha}\right)^{-j}V_{j}(S^{+})
  \quad \text{and}
\quad V_{j}(S^{+}) \geq \left(1 + \frac{1}{\alpha}\right)^{-j}V_{j}(S^{-}).$$
\end{lemma}
\begin{proof}
We follow the argument in the proof of Lemma \ref{lem:splitratio}. The only
difference is that we now must consider the dilation of the cone $C$ by a factor
of $1 + \frac{1}{\alpha}$ about $q$, as the ratio of the width of $S_t$ to the
width of $S^{-}$ (or $S^{+}$) is at most $1 + \frac{1}{\alpha}$.  \end{proof}

\begin{corollary}\label{cor:splitvol}
If $p_t$ is chosen so that $\Vol(S^{+}) = \Vol(S^{-})$, then for all $1 \leq j \leq d$,

$$V_{j}(S^{-}) \geq \left(1 + \frac{1}{\alpha}\right)^{-j}V_{j}(S^{+})
  \quad \text{and} \quad
  V_{j}(S^{+}) \geq \left(1 + \frac{1}{\alpha}\right)^{-j}V_{j}(S^{-}),$$

\noindent
where $\alpha = 2^{1/d} - 1 = \Theta(d^{-1})$.
\end{corollary}
\begin{proof}
Follows from Lemmas \ref{lem:splitvol1} and \ref{lem:splitvol2}.
\end{proof}

\begin{theorem}\label{thm:half_volume}
The algorithm that always sets $p_t$ such that $\Vol(S^+_t) = \Vol(S^-_t)$
has regret bounded by $2^{O(d\log d)}$ for the symmetric loss.
\end{theorem}

\begin{proof}
We will proceed similarly to the analysis of the SymmetricSearch algorithm. Consider a fixed round $t$, let $w = \frac{1}{2}\wid(S_t;u_t)$, and let $K$ be the intersection of the hyperplane $\langle x, u_t\rangle = p_t$ with $S_t$. 

Let $\alpha = 2^{1/d} - 1$, and let $\lambda = 1 + \frac{1}{\alpha}$. Note that $\lambda \geq 2$ and $\lambda = \Theta(d)$. Define a sequence of constants $c_i$ so that $c_0 = 1$ and $c_i/c_{i-1} = \lambda^{-(i+1)}/i$ (in other words, $c_i = \lambda^{-(i+1)(i+2)/2}/i!$). For $1 \leq i \leq d$, define $L_i = (V_{i}(K)/c_i)^{1/i}$, and let $L_0 = \infty$. Choose $j$ so that $L_{j-1} \geq w \geq L_j$ (it is always possible to do this by the same logic as in Theorem \ref{thm:symsearch}). 
The proof of Lemma \ref{lemma:symm1} again carries over verbatim to show that
  $V_{j}(S_t) \geq \frac{1}{j}c_{j-1}w^j$. We again proceed in two steps
  similarly to the proof of Theorem \ref{thm:width_half}.\\

\emph{Step 1:} We first show that  $V_{j}(S_{t+1}) \leq (1 - \lambda^{-(j+2)})V_{j}(S_t)$. 

  The set $S_{t+1}$ is either $S^{+}$ or $S^{-}$. By Corollary \ref{cor:splitvol}, we therefore have that $V_{j}(S_{t+1}) \leq \frac{1}{1+\lambda^{-j}}(V_{j}(S^{+}) + V_{j}(S^{-}))$.

Now, since $V_j$ is a valuation, $V_{j}(S^{+}) + V_{j}(S^{-}) = V_{j}(S_t) + V_{j}(K)$. Since $w \geq L_j$, $V_{j}(K) \leq c_jw^{j}$. Combining this with the fact $V_{j}(S_t) \geq \frac{1}{j}c_{j-1}w^j$, we observe that $V_{j}(K) \leq j\frac{c_j}{c_{j-1}}V_{j}(S_t)$, and therefore that

$$V_{j}(S^+) + V_j(S^{-}) \leq \left(1 + j\frac{c_{j}}{c_{j-1}}\right)V_j(S_t) = (1 + \lambda^{-(j+1)})V_{j}(S_t).$$

It follows that (since $\lambda \geq 2$)

$$V_{j}(S_{t+1}) \leq \frac{1}{1+\lambda^{-j}} \cdot (1 + \lambda^{-(j+1)}) V_{j}(S_t) \leq (1 - \lambda^{-(j+2)})V_{j}(S_t).$$

  \emph{Step 2:} We next consider the potential function 
  $\Phi(t) = \sum_{i=1}^{d} V_{i}(S_t)^{1/i}$. Note that $\Phi(0) =
  \mathrm{poly}(d)$. We will show that each round, $\Phi(t)$ decreases by at
  least $2^{-O(d\log d)}w$. Since the loss each round is upper bounded by $w$,
  this proves our theorem.

In particular, note that

\begin{eqnarray*}
V_{j}(S_{t})^{1/j} - V_{j}(S_{t+1})^{1/j} &\geq & (1 - (1-\lambda^{-(j+2)})^{1/j})V_{j}(S_t)^{1/j} 
\geq  \frac{1}{j}\lambda^{-(j+2)}V_{j}(S_t)^{1/j} \\
&\geq & \frac{1}{j}\lambda^{-(j+2)}\left(\frac{c_{j-1}}{j}\right)^{1/j}w 
\geq  2^{-O(d\log d)}w.
\end{eqnarray*}
\end{proof}




\section{General loss functions}\label{sect:general_loss}

Throughout this paper we have focused on the special cases of the symmetric loss function and the pricing loss function. In this subsection we briefly explore the landscape of other possible loss functions and what regret bounds we can obtain for them. 

For simplicity, we restrict ourselves to loss functions of the form
$\ell(\langle u_t, v\rangle, p_t) = F(\langle u_t, v \rangle - p_t)$. Note that
while some functions (e.g. the pricing loss function) may not be of this form,
they may be dominated by some function of this form (e.g. $F(x) = x$ for $x \geq
0$ and $F(x) = 1$ for $x \leq 0$), and hence any regret bound
that holds for this simplified loss function holds for the original loss
function.

We begin by showing that if $F(x)$ goes to $0$ polynomially quickly from both sides (i.e. if $F(x) \leq |x|^{\beta}$ for some $\beta > 0$), then SymmetricSearch still achieves constant regret.

\begin{theorem}
If $F(x) = |x|^{\beta}$, for $\beta > 0$, then SymmetricSearch (Algorithm \ref{algo:symmetricdd}) achieves regret $O_{d, \beta}(1)$ for the contextual search problem with this loss function. 
\end{theorem}
\begin{proof}
We modify the proof of Theorem \ref{thm:symsearch} to look at the potential function $\Phi_t = \sum_{i=1}^{d} V_{i}(S_t)^{\beta/i}$. The change in potential in each round is now at least (for some $j \in [d]$):

\begin{eqnarray*}
V_{j}(S_{t})^{\beta/j} - V_{j}(S_{t+1})^{\beta/j} &\geq & \left(1 - \left(\frac{3}{4}\right)^{\beta/j}\right)V_{j}(S_t)^{\beta/j} \\
&\geq & \left(1 - \left(\frac{3}{4}\right)^{\beta/j}\right)\left(\frac{c_{j-1}}{j}\right)^{\beta/j} w^{\beta}  \geq O_{d,\beta}(1) \ell_{t}.
\end{eqnarray*}

It follows that the total regret of SymmetricSearch is $O_{d,\beta}(1)$.
\end{proof}

Similarly, we can show that for functions $F$ which are discontinuous on one side and converge to zero polynomially quickly on the other side, the PricingSearch algorithm (with a slightly different choice of parameters) achieves $O_{d, \alpha}(1)$ regret.

\begin{theorem}
Let $\alpha > 0$ be a constant, and let $F(x) = |x|^{\beta}$, for $x \geq 0$ and let $F(x) = 1$ for $x < 0$. PricingSearch (Algorithm \ref{algo:pricingdd}) with parameter $\alpha = 1 + \frac{\beta}{d}$ achieves regret $O_{d, \beta}(\log \log T)$ for the contextual search problem with this loss function. 
\end{theorem}
\begin{proof}
Again, we modify the proof of Theorem \ref{thm:pricing}. Lemmas \ref{lem:pricing1}, \ref{lemma:pricing2}, \ref{lemma:pricing3}, and \ref{lemma:pricing4} hold as written. The only necessary change is in the underpricing case of the proof of Theorem \ref{thm:pricing}, where the maximum loss for an event is now $(4\ell_{k_J})^{\beta}$, and so the total loss from underpricing (for a fixed value of $J$ and $k_J$) is at most

  \begin{eqnarray*}
  \frac{2 \ell^J_{k_J} }{\ell^J_{k_J+1}} \cdot (4 \ell_{k_J})^{\beta}
  &=& 2^{1+2\beta} d^{2\beta} \exp(J \alpha^{k_J +1} - (J+\beta) \alpha^{k_J})\\
  &\leq & 2^{1+2\beta} d^{2\beta} \exp(\alpha^{k_J} (d \alpha - (d+\beta))\\
    &=& 2^{1+2\beta} d^{2\beta} = O_{d,\beta}(1).
  \end{eqnarray*}
  
It follows that the total regret of PricingSearch is $O_{d,\beta}(\log\log T)$.
\end{proof}

\bibliographystyle{plain}
\footnotesize
\bibliography{pricing}
\normalsize

\appendix

\section{Analysis of the $1$-dimensional case}\label{appendix:one-dim}

We now analyze the policy of Kleinberg and Leighton for the one-dimensional
case. Their keep a knowledge set $S_t = [a_t, a_t + \Delta_t]$ and
choose price $$p_t = a_t + 1/2^{2^{k_t}} \quad \text{where} \quad k_t = \lfloor 1+\log_2 \log_2
\Delta_t^{-1} \rfloor$$ while $\Delta_t > 1/T$ after that, their policy prices at the lower end
of the interval.

Clearly the total regret whenever $\Delta_t \leq 1/T$ is at most $1$, so we only
need to analyze the cases where $\Delta_t > 1/T$ and hence $k_t \leq O(\log \log
T)$. To show a regret bound of $O(\log \log T)$ it is enough to argue that for
every value of $k$, the total regret from timesteps where $k_t = k$ is $O(1)$. 

We start by noting that if there is no sale then in the next period
$\Delta_{t+1} = 1/2^{2^{k_t}}$ and therefore $k_{t+1} = k_t + 1$. Since $k_t$ is
monotone, there can be at most one no-sale for every value of $k_t$. The
remaining periods where $k_t = k$ correspond to sales, where the loss is at most
$\Delta_t \leq 1/2^{2^{k-1}}$, since by each sale $\Delta_t$ decreases by
$1/2^{2^{k}}$, there are at most $2^{2^{k}} / 2^{2^{k-1}} = 2^{2^{k-1}}$
sales. Since each of them incur loss $\Delta_t \leq 1/2^{2^{k-1}}$, the total
regret for sales with $k_t = k$ is at most $1$. The total regret from no-sales
is at most $1$ since there is at most one no-sale.

\end{document}